\documentclass[fleqn,10pt]{wlscirep}
\usepackage[utf8]{inputenc}
\usepackage{amsthm}
\usepackage[T1]{fontenc}
\usepackage{bm}

\usepackage{MnSymbol}

\usepackage{tikz-cd} 
\usepackage{quiver}

\setlist[enumerate]{leftmargin=*}
\setlist[itemize]{leftmargin=*}

\newtheorem{theorem}{Theorem}
\newtheorem{lemma}{Lemma}

\newtheorem{definition}{Definition}

\title{Spatial dynamics formulation of magnetohydrostatics}

\author[1,*]{J. W. Burby}
\author[2,**]{M. H. Updike}
\affil[1]{Los Alamos National Laboratory, Theoretical Division, Los Alamos, NM 87545, USA}
\affil[2]{University of Texas at Austin, Austin, TX 78712, USA}

\affil[*]{jburby@lanl.gov}
\affil[**]{michaelupdike@utexas.edu}

\begin{abstract}
We present a formalism for importing techniques from dynamical systems theory into the study of three-dimensional static ideal magnetohydrodynamic (MHD) equilibria in toroidal domains. By treating some choice of toroidal angle as time, we reformulate the force balance and divergence equations as a system of hydrodynamic equations on the unit disc. We show these spatial dynamics equations satisfy a variational principle and comprise a Lie-Poisson Hamiltonian system on the dual to a certain infinite-dimensional Lie algebra. We use the variational principle to find Noether conservation laws, including a general circulation conservation law, a vorticity advection law when pressure vanishes, and an energy conservation law in axisymmetric domains. Combining the Lie-Poisson structure with a construction due to Scovel-Weinstein, we develop a theory of smoothed particle magnetohydrostatics (SPMHS). SPMHS identifies a large class of exact particle-swarm solutions of regularized spatial dynamics equations. The regularization occurs at finer scales than the smallest scales within the physical purview of ideal MHD. Crucially, the SPMHS equations of motion comprise a finite-dimensional Hamiltonian system, which side steps perennial roadblocks to a satisfactory theory of three-dimensional equilibria. In large-aspect-ratio domains, we show that the spatial dynamics equations comprise a fast-slow system, where fast dynamics corresponds to the elliptic part of the equilibrium equations and slow dynamics corresponds to the hyperbolic part. Because the fast dynamics is formally normally hyperbolic, Fenichel theory suggests the presence of an exact slow manifold for the spatial dynamics equations that contains all physical solutions. This is analogous to Spohn's resolution of the paradoxical nature of the Lorentz-Abraham-Dirac equation for a single electron experiencing its self force. We formally carry out reduction to the slow manifold, obtaining a spatial dynamics formulation for equilibrium Strauss reduced MHD at leading order and a hierarchy of corrections at any desired order in perturbation theory. In large aspect-ratio domains, finding periodic solutions of the slow manifold reduced spatial dynamics equations represents a second novel pathway to finding three-dimensional equilibria.
\end{abstract}
\begin{document}

\flushbottom
\maketitle
%
%
\thispagestyle{empty}






\section*{Introduction}
Toroidal magnetic confinement fusion (MCF) uses strong magnets to keep hot burning plasma inside of a compact spatial region $Q$ topologically equivalent to a donut. Two main approaches dominate the current research landscape devoted to making toroidal MCF power reactors an economic reality: the tokamak and the stellarator. Tokamaks entail a region $Q$ that is symmetric under arbitrary rotations about a fixed axis, while stellarator $Q$'s exhibit great geometric complexity. In either case, the magnetic field $\bm{B}$ inside of the toroidal chamber $Q$ satisfies the equations of magnetohydrostatics (MHS)
\begin{align}
(\nabla\times\bm{B})\times\bm{B} = \nabla p ,\quad \nabla\cdot \bm{B}=0,\quad \bm{B}\cdot \bm{n}=0\text{ on }\partial Q,\label{mhs}
\end{align}
under standard quiescent operating conditions. Here $p$ denotes the plasma pressure and $\bm{n}$ is the outward-pointing unit normal along $\partial Q$. Much of the current understanding of desireable toroidal MCF reactor concepts relies on theory of the partial differential equations (PDE) \eqref{mhs}, and their implications for particle confinement and plasma stability. While MHS magnetic fields in tokamaks are fairly well understood, stellarator MHS fields still present many theoretical and practical challenges.

The solutions of \eqref{mhs} relevant to tokamaks admit a continuous rotation symmetry. Standard manipulations show that for such symmetric solutions Eqs.\,\eqref{mhs} simplify to a single elliptic PDE for an unknown stream function $\psi$. The close relationship between the PDE for $\psi$, known as the Grad-Shafranov equation, and Poisson's equation leads to a fairly straightforward theory of symmetric MHS solutions. 

Solutions of \eqref{mhs} relevant to stellarators generally do not admit any continuous rigid symmetry. Instead they admit hidden volume-preserving symmetries\cite{BKM_2020a}.  These hidden symmetries ensure the existence of a stream function $\psi$ that satisfies a generalization of the Grad-Shafranov equation. However, this generalized Grad-Shafranov equation is only a necessary condition for equilibrium, in stark contrast to the case of rigid symmetries, where the Grad-Shafranov equation is both necessary and sufficient for equilibrium. 

The theory of non-rigidly-symmetric MHS solutions is still quite limited, in spite of the considerable importance of that theory to MCF. The simplest solutions are probably the vacuum fields, characterized by $\nabla \times\bm{B} = 0$ and $p=\text{const}$. The existence and uniqueness theory is straightforward in this case because the system $\nabla\cdot \bm{B} = 0$, $\nabla\times\bm{B} = 0$ is elliptic. However, vacuum fields in general possess poor particle confinement properties owing to lack of a non-trivial stream function $\psi$ that is constant along $\bm{B}$-lines. The rigorous theory of vacuum fields with stream functions is undeveloped. The next simplest solutions are the so-called Beltrami fields, which satisfy $\nabla\times\bm{B} = \lambda\,\bm{B}$ for some constant $\lambda$. Many rigorous techniques from the theory of elliptic partial differential equations apply in the Beltrami case, but again there is no guarantee of existence of a stream function, and particle confinement properties are generically poor. Moving away from special classes of solutions that are really governed by linear partial differential equations, the earliest rigorous result is due to Lortz\cite{Lortz_1970,Weitzner_2020}, who proved existence of smooth MHS solutions with non-constant pressure in arbitrary toroidal domains $Q$ that are symmetric under reflections through a plane. The Lortz solutions are remarkable because they do not in general admit a continuous family of rigid symmetries. But all of these solutions possess a serious practical shortcoming: they lack so-called rotational transform, which is the tendancy of $\bm{B}$-lines to wind helically around a central elliptic periodic orbit for $\bm{B}$. Without rotational transform, particle confinement is known to be poor. Then comes the three-dimensional solutions with non-constant pressure\cite{BL_1996} due to Bruno and Laurence. These non-smooth solutions are Beltrami in each of a nested sequence of toroidal annuli. The Beltrami constant $\lambda$ differs in each region, and non-trivial interface conditions determine the jump discontinuities in $\bm{B}$ across the toroidal interfaces, thereby making a strong connection with the method underlying P. Garabedian's NSTAB\cite{Garabedian_2008} equilibrium and stability code. The richness and mathematical rigor underlying Bruno-Laurence solutions, independently studied by Hole, Hudson, and Dewar\cite{Hole_2007}, motivated the development of a novel practical approach to computation of MHS solutions implemented in SPEC\cite{Hudson_2012}. But lack of smoothness introduces complications in the foundations of MCF theory that still need to be adequately addressed. Moreover, the rigorous Bruno-Laurence theory requires small deviations from axisymmetry\cite{Qu_2021} -- an unfortunate limitation in the stellarator context. Finally, and most recently, Constantin-Pasqualotto\cite{Constantin_2023} constructed topologically non-trivial smooth solutions in arbitrary domains with non-constant pressure. These solutions are new, and still require further analysis to determine their relevance to MCF. For example, Constantin-Pasqualotto do not rule out the possibility that the obtained solutions comprise a core non-vacuum region that is axisymmetric surrounded by a non-axisymmetric vacuum region. Overall, the theory of smooth, three-dimensional MHS solutions with non-constant pressure is largely underdeveloped. New methods need to be developed even to settle the question of existence of such solutions. Such methods will likely inspire new computational methods as well.

This Article introduces an unexplored approach to studying the MHS equations known as \textbf{spatial dynamics}\cite{Groves_Schneider_2001,Groves_Schneider_2005,Groves_Schneider_2008,Schneider_Uecker_2017}. The basic idea behind spatial dynamics is to reformulate a given PDE system as a system of evolutionary PDEs by identifying some spatial coordinate with ``time". After reformulating a PDE as a spatial dynamical system, techniques from dynamical systems theory become available for analyzing solutions of the original equations. Recent advances along these lines in the context of elliptic PDE are descrbed by M. Beck\cite{Beck_2020}. 

Let $(\bm{B},p)$ be a solution of the MHS equations in a domain $Q$ diffeomorphic to a solid torus, $Q\approx S^1\times D\ni(\zeta,x,y)$. Here $D\subset \mathbb{R}^2$ denotes the standard unit disc and $S^1 = \mathbb{R}\text{ mod }2\pi$. Assume that the normal component of $\bm{B}$ along the boundary $\partial Q$ vanishes and that the toroidal component of the magnetic field $ B^\zeta =\bm{B}\cdot \nabla\zeta$ is nowhere vanishing. I will show formally that such $(\bm{B},p)$ arise as $2\pi$-periodic orbits for a spatial dynamics reformulation of MHS with the following properties. The spatial dynamics equations:
\begin{enumerate}
\item[(1)] describe evolution of MHS solutions through ``time", where the toroidal angle $\zeta$ plays the role of time,
\item[(2)] take the form of a system of planar, time-dependent hydrodynamic equations,
\item[(3)] comprise a Lie-Poisson Hamiltonian system on the dual to a certain infinite-dimensional Lie algebra.
\end{enumerate}

As a consequence of its curious mathematical structure, the spatial dynamics formulation of MHS leads to at least two new ways to study three-dimensional equilibria. First, in conjunction with the Lie-algebraic theory due to Scovel-Weinstein\cite{SW_1994}, the Lie-Poisson formulation of spatial dynamics leads naturally to a variant of the smoothed-particle hydrodynamics\cite{Monaghan_1992} method (SPH) for MHS. We show that solutions of the smoothed-particle MHS (SPMHS) equations correspond to exact solutions of regularized spatial dynamics equations. Since the regularization happens at length scales beyond the purview of MHD, the regularized equations comprise a physically-viable substitute for the usual MHS equations. Thus, SPMHS enables the study of stellarator configurations by searching for periodic solutions of a finite-dimensional Hamiltonian system. Finite-dimensionality side-steps the traditional obstructions to a good existence theory for MHS that underlie Grad's famous conjecture\cite{Grad_1967}. (It is interesting to note that Monaghan's review\cite{Monaghan_1992} of SPH begins by highlighting SPH was invented to simulate non-axisymmetric astrophysical phenomena.) Second, for large-aspect-ratio domains, MHS spatial dynamics reveals that any solution of the MHS equations must lie on an infinite-dimensional slow manifold. The elliptic part of the MHS equations govern the fast dynamics normal to the slow manifold, while the hyperbolic part of the MHS equations governs the slower motion along the slow manifold\cite{Fenichel_1979,Lorenz_1986,Lorenz_1987,Lorenz_1992,MacKay_2004,Burby_Klotz_2020}. Formally reducing the spatial dynamics equations to the slow manifold ``solves" the elliptic part of the MHS problem, leaving simpler hyperbolic equations to analyze; at leading-order the slow dynamics recover the vorticity form of the planar Euler equations, or, equivalently, the equilibrium conditions for Strauss'\cite{Strauss_1975} reduced MHD. This observation represents an opportunity to make progress in stellarator equilibrium theory because an important part of the difficulty the MHS equations present is their mixed elliptic-hyperbolic type. 

The paper is organized as follows. Section \ref{sec:direct} derives the MHS spatial dynamics equations directly from the usual boundary-value formulation of MHS. This Section introduces much notation and terminology used in the rest of the Article, in particular with relation to toroidal coordinates. Section \ref{variational_formulation} demonstrates the the spatial dynamics equations from Section \ref{sec:direct} satisfy a variational principle, and Section \ref{cons_noether} combines that variational principle Noether's theorem to deduce several conservation properties of the spatial dynamics equations. The variational formulation suggests there should be a concommittant Hamiltonian formulation. Section \ref{sec:lie_poisson} shows that in fact the spatial dynamics equations comprise a Lie-Poisson system on the dual to a Lie algebra. The ``time" evolution perspective of MHS suggests considering spatial domains that exhibit multiple ``timescales." Section \ref{sec:fs_formulation} reveals that the spatial dynamics in a large-aspect-ratio domain comprise a fast-slow system for which Fenichel's geometric singular perturbation theory formally applies. Finally, Section \ref{sec:snakes} combines the Lie-Poisson structure identified in Section \ref{sec:lie_poisson} with a Lie algebra reduction technique due to Scovel-Weinstein to formulate the theory of smoothed-particle magnetohydrostatics. Section \ref{sec:discussion} contains a forward looking discussion on implications of the spatial dynamics approach to MHS.

Whenever possible, we present each of our key results using both standard index notation, where repeated indices imply summation, and differential forms. When working with differential forms, $\mathcal{L}_{\bm{w}}$ always denotes Lie differentiation along the vector field $\bm{w}$; $d$ denotes exterior differentiation; and $\iota_{\bm{w}}$ denotes interior product with the vector field $\bm{w}$. We reserve the symbol $\bm{n}$ for the outward-pointing unit normal vector along the boundary of a spatial domain.

\section{Direct derivation of MHS spatial dynamics\label{sec:direct}}
In order to reformulate the boundary value problem \eqref{mhs} as a time evolution problem, the first step is defining "time". Slice the toroidal spatial domain $Q$ into hypersurfaces, each diffeomorphic to the standard unit disc $D\subset \mathbb{R}^2$. These embedded discs define surfaces of constant time. The time variable itself is therefore an angular coordinate $\zeta:Q\rightarrow\mathbb{S}^1$ (here $S^1 = \mathbb{R}\text{ mod }2\pi$ denotes the circle) that is constant along each hypersurface. The embedded disc whose assigned time is $\zeta\in S^1$ is $D_\zeta$. In keeping with standard conventions from plasma physics, we will refer to $\zeta$ as the toroidal angle on $Q$.

The above notion of time is problematic. In dynamical systems theory, time extends forever in each direction. But the angular coordinate $\zeta$ takes values in the circle $S^1$, a compact set. Since spatial dynamics aims to bring tools from dynamical systems theory into the study of boundary value problems, this problem needs a resolution. Our solution draws from the theory of covering spaces.

The (universal) covering space $\widetilde{Q}$ for the spatial domain $Q$ is an infinitely-long (solid) cylinder together with a prescription for wrapping $\widetilde{Q}$ around $Q$. Formally, the wrapping prescription is a smooth map $\mathcal{P}:\widetilde{Q}\rightarrow Q$ with the following properties.
\begin{itemize}
\item $\mathcal{P}$ is surjective
\item For each $\widetilde{q}\in\widetilde{Q}$ the tangent map $T_{\widetilde{q}}\mathcal{P}:T_{\widetilde{q}}\widetilde{Q}\rightarrow T_{\mathcal{P}(\widetilde{q})}Q$ is a linear isomorphism
\end{itemize}
By the inverse function theorem, for each $\widetilde{q}\in \widetilde{Q}$ there is an open neighborhood $\mathcal{U}_{\widetilde{q}}$ containing $\widetilde{q}$ such that $\mathcal{P}\mid \mathcal{U}_{\widetilde{q}}$ is a diffeomorphism onto its image. Thus, locally, the covering space $\widetilde{Q}$ looks like $Q$. But globally $\widetilde{Q}$ "unwraps" the periodicity of the solid torus $Q$. Note in particular that $\widetilde{Q}$ is simply connected.

The covering space $\widetilde{Q}$ leads to an improved notion of time as follows. The differential $d\zeta$ defines an ordinary closed $1$-form on $Q$. The pullback $\mathcal{P}^*d\zeta$ is therefore a closed $1$-form on the covering space $\widetilde{Q}$. But every closed $1$-form on $\widetilde{Q}$ is exact because $\widetilde{Q}$ is simply connected. In other words, there is a smooth function $\widetilde{\zeta}:\widetilde{Q}\rightarrow \mathbb{R}$, defined modulo addition of constants, such that $d\widetilde{\zeta} = \mathcal{P}^*d\zeta$. Locally, $\widetilde{\zeta}$ behaves like $\zeta$; level sets of $\widetilde{\zeta}$ foliate covering space by hypersurfaces diffeomorphic to discs, and this foliation maps along $\mathcal{P}$ onto the hypersurface foliation of $Q$. Globally, values of $\widetilde{\zeta}$ extend forever in each direction. The toroidal coordinate $\widetilde{\zeta}$ on covering space defines "time" in the classical sense used in dynamical systems theory. From here on, the same symbol $\zeta$ will be used for both the toroidal angle on $Q$ and the toroidal angle on $\widetilde{Q}$.

The notion of time defined on the covering space is useful because study of the MHS boundary value problem \eqref{mhs} on $Q$ may be replaced with study of the MHS boundary value problem on $\widetilde{Q}$. This might sound strange as we have not yet defined an MHS boundary value problem on $\widetilde{Q}$. But since $\widetilde{Q}$ and $Q$ are locally the same space, and the MHS equations on $Q$ are PDEs, the MHS boundary value problem on $Q$ induces a corresponding boundary value problem on covering space in a unique way. The metric tensor $\widetilde{g}$ on $\widetilde{Q}$ is given by pulling back the metric tensor $g$ on $Q$ along $\mathcal{P}$, $\widetilde{g} = \mathcal{P}^*g$. The vector calculus operations, divergence, gradient, and curl, are then defined on $\widetilde{Q}$ in terms of $\widetilde{g}$ in the usual way (CITE). The natural MHS boundary value problem on $\widetilde{Q}$ is therefore
\begin{align}
(\nabla\times\bm{B})\times\bm{B} = \nabla p ,\quad \nabla\cdot \bm{B}=0,\quad \bm{B}\cdot \bm{n}=0\text{ on }\partial \widetilde{Q},\label{mhs*}
\end{align}
where $\bm{B}$ is a vector field on $\widetilde{Q}$, $p$ is a function on $\widetilde{Q}$, and the derivative operators are defined using $\widetilde{g}$. The following technical lemma establishes the precise relationship between solutions of the MHS boundary value problems on $Q$ and $\widetilde{Q}$. It justifies leaving the MHS boundary value problem on $Q$ behind, and embracing the MHS boundary value problem on $\widetilde{Q}$ going forward. Readers unfamiliar with deck transformations may skip this result without losing the narrative thread. We will re-express this result in more elementary terms after the proof of Theorem \ref{sdf}.

\begin{lemma}\label{deck_transformations_lemma}
Solutions of the MHS boundary value problem \eqref{mhs} on $Q$ are in one-to-one correspondence with solutions of the MHS boundary value problem \eqref{mhs*} on $\widetilde{Q}$ that are invariant under deck transformations.
\end{lemma}

By now we have defined the time variable $\zeta$ for a spatial dynamics formulation of MHS. Next we need to define the "space" over which our dynamical fields will be defined. The points in space at a given time clearly correspond to a level set of $\zeta$. But a proper dynamical systems formulation requires an enduring notion of space --- not a different space for each time. For this we introduce toroidal coordinates as follows.

\begin{definition}
A submanifold with boundary $Q\subset \mathbb{R}^3$ is a \textbf{toroidal domain} if there is an orientation-preserving diffeomorphism $\varphi:Q\rightarrow S^1\times D$, where $D\ni (x,y)$ denotes the standard unit disc in $\mathbb{R}^2$, $S^1=  \mathbb{R}\text{ mod }2\pi\ni \zeta$, and the orientation on $S^1\times D$ is determined by the volume form $d\zeta\wedge dx\wedge dy$. The component functions of the diffeomorphism will be denoted $\varphi = (\zeta,x,y)$. The diffeomorphism itself will be referred to as \textbf{toroidal coordinates}. The angular component $\zeta$ is the \textbf{toroidal angle}. If $\widetilde{Q}$ denotes the universal cover for a toroidal domain $Q$ then there is a natural orientation-preserving coordinate system (defined modulo deck transformations) $\widetilde{\varphi}:\widetilde{Q}\rightarrow \mathbb{R}\times D$. The components of $\widetilde{\varphi} = (\zeta,x,y)$ will be denoted using the same symbols representing $\varphi$. For $\zeta_1,\zeta_2$ real constants with $\zeta_1 \leq \zeta_2$, the symbol $\widetilde{Q}(\zeta_1,\zeta_2)$ denotes the \textbf{toroidal interval} $\widetilde{\varphi}^{-1}([\zeta_1,\zeta_2]\times D)$. Similarly $D_{\zeta_1} = \widetilde{\varphi}^{-1}( \{\zeta_1\}\times D)$ denotes a \textbf{toroidal slice}.
\end{definition}

Given a toroidal coordinate system $\widetilde{\varphi} = (\zeta,x,y)$ on $\widetilde{Q}$, each point $\widetilde{q}\in \widetilde{Q}$ is assigned a unique "time" $\zeta(\widetilde{q})\in\mathbb{R}$ and a unique point in "space" $(x(\widetilde{q}),y(\widetilde{q}))\in D$. In this manner, a scalar field on $\widetilde{Q}$, like pressure $p(\widetilde{q})$, inherits a new identity. Instead of a static function of $3$ variables, it becomes a dynamical function of $2$ variables: $p_\zeta(x,y) = p(\widetilde{\varphi}^{-1}(\zeta,x,y))$. In order to complete the spatial dynamics formulation of MHS, the magnetic field $\bm{B}$ requires an analogous dynamical interpretation, and the PDE system \eqref{mhs*} needs to be expressed in a way that distinguishes between space and time derivatives. These objectives may be achieved in a variety of ways. But we will proceed in a manner that efficiently encodes the geometry of the hypersurface foliation.

Hypersurface foliations figure prominently in the famous ADM formalism for general relativity (CITE). That formalism motivates introducing the following geometrical quantities.
\begin{definition}
Let $Q$ be a toroidal domain with universal cover $\widetilde{Q}$. Let $\varphi:Q\rightarrow S^1\times D$ and $\widetilde{\varphi}:\widetilde{Q}\rightarrow \mathbb{R}\times D$ denote the corresponding toroidal coordinate systems $(\zeta,x,y)$. Fix an arbitrary $\zeta\in\mathbb{R}\text{ or }S^1$. The \textbf{poloidal inclusion map} $I_\zeta:D\rightarrow Q$ or $I_\zeta:D\rightarrow \widetilde{Q}$ is given by $I_\zeta(x,y) = \varphi^{-1}(\zeta,x,y)$ or $I_\zeta(x,y) = \widetilde{\varphi}^{-1}(\zeta,x,y)$.
The \textbf{hypersurface metric} is the metric tensor $h_\zeta$ on $D$ given by $h_\zeta = I_\zeta^*\widetilde{g}$. The \textbf{shift form} $\mathcal{N}_\zeta$ is the $1$-form on $D$ given by $\mathcal{N}_\zeta = I_{\zeta}^*(\iota_{\partial_\zeta}\widetilde{g})$. The \textbf{shift vector} $\bm{N}_\zeta$ is the $h_\zeta$ metric sharp of the shift form, i.e. the unique vector field on $D$ that satisfies $\iota_{\bm{N}_\zeta}h_\zeta = \mathcal{N}_{\zeta}$. The \textbf{lapse function} $N_\zeta$ is the scalar field on $D$ given by $N_\zeta^2 = I_\zeta^*(\widetilde{g}(\partial_\zeta,\partial_\zeta)) - h_\zeta(\bm{N}_\zeta,\bm{N}_\zeta).$ The \textbf{hypersurface area element} $\omega_\zeta$ is the nowhere-vanishing $2$-form on $D$ given by $\omega_\zeta = I_\zeta^*(\iota_{\partial_\zeta}\Omega)$, where $\Omega$ denotes the Riemannian volume form on $\widetilde{Q}$. The \textbf{Riemannian hypersurface area element} is the nowhere-vanishing $2$-form $\sigma_\zeta$ on $D$ associated with the hypersurface metric $h_\zeta$: $\sigma_\zeta= \sqrt{\text{det}\,h_\zeta}\,dx\wedge dy$.
\end{definition}

The image of the poloidal inclusion map, $D_\zeta = I_\zeta(D)$, coincides with the embedded disc whose assigned time is $\zeta$. The map $I_\zeta$ itself therefore parameterizes the time-$\zeta$ embedded disc $D_\zeta$. As for any embedded surface, the disc $D_\zeta$ inherits notions of length and angle from the ambient Euclidean metric $\widetilde{g}$ on covering space. In other words, $D_\zeta$ possesses its own intrinsic Riemannian metric. The hypersurface metric $h_\zeta$ on $D$ simply expresses the metric on $D_\zeta$ using the parameterization provided by the poloidal inclusion map $I_\zeta$. The shift form $\mathcal{N}_\zeta$ encodes the mixed space-time components of the Euclidean metric $\widetilde{g}$ written in toroidal coordinates. The lapse function $N_\zeta$ gives the normal velocity of $D_\zeta$, regarded as a time-evolving surface. The hypersurface area element $\omega_\zeta$ represents the Lebesgue measure on covering space conditioned on the toroidal angle $\zeta$. This notion of area differs from the intrinsic, Riemannian notion of area on $D_\zeta$. In fact, $\omega_\zeta$ and $\sigma_\zeta$ are related by the lapse function according to $\omega_\zeta = N_\zeta\,\sigma_\zeta$.

The above geometrical objects attached to a hypersurface foliation suggest a convenient set of dependent variables for expressing the spatial dynamics formulation of MHS. We restrict attention to \textbf{toroidal magnetic fields} $\bm{B}$, characterized by the property $d\zeta(\bm{B}) > 0$. The pressure in a hypersurface is $p_\zeta = I_\zeta^*p$, a function on $D$. The toroidal flux is $\varrho_\zeta = I_\zeta^*(\iota_{\bm{B}}\Omega)$, a nowhere-vanishing $2$-form on $D$. The field-line velocity is the unique vector field $\bm{u}_\zeta$ on $D$ such that $\iota_{\bm{u}_\zeta}\rho_\zeta = - I_\zeta^*(\iota_{\partial_\zeta}\iota_{\bm{B}}\Omega)$. These three time-dependent quantities on $D$ completely encode $\bm{B}$ and $p$, as the following lemma shows.

\begin{lemma}\label{B_param_lemma}
Let $\bm{B}$ be a toroidal magnetic field, $p$ a pressure function, and $(\zeta,x,y)$ a system of toroidal coordinates. If $p_\zeta,\varrho_\zeta,\bm{u}_\zeta$ denote the corresponding hypersurface pressure, toroidal flux, and field line velocity then
\begin{align*}
\bm{B}(\zeta,x,y) = \frac{\varrho_\zeta(x,y)}{\omega_\zeta(x,y)}\bigg(u^x_\zeta(x,y)\,\partial_x + u^y_\zeta(x,y)\,\partial_y + \partial_\zeta\bigg),\quad p(\zeta,x,y)  = p_\zeta(x,y),
\end{align*}
where $\omega_\zeta$ denotes the hypersurface area form.

\end{lemma}
\begin{proof}
    It is obvious that $p(\zeta,x,y) = I_\zeta^*p(x,y) = p_\zeta(x,y)$. So to prove this statement, we first note that there exists functions $n(\zeta,x,y)$, $B^x_\zeta(x,y)$, and $B^y_\zeta(x,y)$ such that 
    \[
    \bm{B}(\zeta,x,y) = n(\zeta,x,y)\partial_{\zeta} +\bm{B}^\perp(\zeta,x,y) ,
    \]
    where $\bm{B}^{\perp}(\zeta,x,y) = B^x_\zeta(x,y)\,\partial_x + B^y_\zeta(x,y)\,\partial_y$ is tangent to the hypersurface $D_{ \zeta}$. 
    By definition of $\bm{B}^\perp$ and anti-commutativity of interior products, for each vector field $\mathbf{v}(\zeta,x,y)= v^x_\zeta(x,y)\,\partial_x + v^y_\zeta(x,y)\,\partial_y$ tangent to the hypersurfaces $D_\zeta$,
    \begin{align*}
    \iota_{\mathbf{v}}\iota_{\bm{B}}\iota_{\partial_\zeta}\Omega=\iota_{\mathbf{v}}\iota_{\bm{B}^\perp}\iota_{\partial_\zeta}\Omega = -\iota_{\mathbf{v}}\iota_{\partial_\zeta}\iota_{\bm{B}}\Omega.
    \end{align*}
    Pulling back this identity along the poloidal inclusion map $I_\zeta$ therefore implies
    \begin{align*}
    \iota_{\mathbf{v}_\zeta}\iota_{\,\bm{B}^\perp_\zeta}\omega_\zeta = -\iota_{\mathbf{v}_\zeta}I_\zeta^*(\iota_{\partial_\zeta}\iota_{\bm{B}}\Omega)=\iota_{\mathbf{v}_\zeta}\iota_{\bm{u}_\zeta}\varrho_\zeta,
    \end{align*}
    where $\bm{v}_\zeta(x,y) = v^x_\zeta(x,y)\,\partial_x + v^y_\zeta(x,y)\,\partial_y$ and $\bm{B}^\perp_\zeta(x,y) = B^x_\zeta(x,y)\,\partial_x + B^y_\zeta(x,y)\,\partial_y$ are $\bm{v}$ and $\bm{B}^\perp$ regarded as vector fields on $D$.
    And since $\bm{v}_\zeta$ is arbitrary we have $\iota_{\bm{B}^\perp_\zeta}\omega_\zeta  =\iota_{\bm{u}_\zeta}\varrho_\zeta$. By the nowhere-vanishing property for $\omega_\zeta$ and $\varrho_\zeta$, we therefore find
    \begin{align*}
    \bm{B}^\perp_\zeta = \frac{\varrho_\zeta}{\omega_\zeta}\bm{u}_\zeta.
    \end{align*}
    
    To finish off the lemma, we now must show that 
    \[
    n(\zeta,x,y) = \frac{\varrho_\zeta(x,y)}{\omega_\zeta(x,y)}.
    \]
    The function $\frac{\varrho_\zeta}{\omega_\zeta}$ can be characterized as the unique function on $D$ such that 
    \[
    \frac{\varrho_\zeta}{\omega_\zeta}\,\omega_\zeta = \varrho_\zeta = I_\zeta^*(\iota_{\bm{B}}\Omega)=I_\zeta^*\left(n\iota_\zeta\Omega\right) = n_\zeta\,\omega_\zeta,\quad n_\zeta = I_\zeta^*n,
    \]
    which is the desired result.
\end{proof}

Each ingredient necessary for deriving the spatial dynamics formulation of MHS is now in place. The following Theorem completes the derivation.

\begin{theorem}\label{sdf}
Let $\bm{B}$ and $p$ be a toroidal vector field and scalar function on $\widetilde{Q}$. The following are equivalent.
\begin{itemize}
\item $(\bm{B},p)$ solves the partial differential equations $(\nabla\times \bm{B})\times\bm{B} = \nabla p$, $\nabla\cdot \bm{B}=0$, subject to the boundary condition $\bm{B}\cdot \bm{n} = 0$ on $\partial \widetilde{Q}$.
\item The toroidal flux $\varrho_\zeta $, the field line velocity $\bm{u}_\zeta$, and the hypersurface pressure $p_\zeta$ on $D$ associated with $(\bm{B},p)$ solve the system of evolution equations
\begin{gather}
\partial_\zeta p_\zeta + \mathcal{L}_{\bm{u}_\zeta}p_\zeta = 0\label{pressure_evo_ld}\\
\partial_\zeta \varrho_\zeta + \mathcal{L}_{u_\zeta}\varrho_\zeta = 0\label{flux_evo_ld}\\
\partial_\zeta\pi_\zeta + \mathcal{L}_{\bm{u}_\zeta}\pi_\zeta   =\frac{dp_\zeta}{n_\zeta} +  d\bigg( \frac{|\pi_\zeta|^2}{n_\zeta} + n_\zeta\,N_\zeta^2  \bigg),\quad \pi_\zeta = n_\zeta(\iota_{\bm{u}_\zeta + \bm{N}_\zeta}h_\zeta),\quad n_\zeta = \frac{\varrho_\zeta}{\omega_\zeta},\label{vel_evo_ld}
\end{gather}
subject to tangential boundary conditions on $\bm{u}_\zeta$: $\bm{u}_\zeta$ is tangent to $\partial D$ for each $\zeta$.

\item The toroidal flux $\varrho_\zeta $, the field line velocity $\bm{u}_\zeta $, and the hypersurface pressure $p_\zeta$ on $D$ associated with $(\bm{B},p)$ solve the system of evolution equations
\begin{gather}
\partial_\zeta p_\zeta +u^i_\zeta\,\partial_i p_\zeta = 0\label{pressure_evo_ind}\\
\partial_\zeta \bigg(n_\zeta\,N_\zeta\,\sqrt{\text{det}\,h_\zeta}\bigg) +\partial_i\bigg(n_\zeta\,N_\zeta\,\sqrt{\text{det}\,h_\zeta}\,u^i_\zeta\bigg) = 0,\quad n_\zeta = \frac{\varrho_\zeta}{\omega_\zeta},\label{flux_evo_ind}\\
\partial_\zeta\pi_{\zeta\,i} + u^j_\zeta\,\partial_j \pi_{\zeta\,i} + \pi_{\zeta\,j}\,\partial_iu^j_\zeta  =\frac{\partial_ip_\zeta}{n_\zeta} +  \partial_i\bigg( \frac{|\pi_\zeta|^2}{n_\zeta} + n_\zeta\,N_\zeta^2  \bigg),\quad \pi_{\zeta\,i} = n_\zeta\,h_{\zeta\,i\,j}(u^j_\zeta + N_\zeta^j) ,\label{vel_evo_ind}
\end{gather}
subject to tangential boundary conditions on $\bm{u}_\zeta$.
\end{itemize}
\end{theorem}
\begin{proof}
The third bullet merely rewrites the second bullet using index notation. So we will prove equivalence of the first two bullets.

First we show that the MHS equation $\nabla\cdot \bm{B} = 0$ is equivalent to the spatial dynamics equation $(\partial_\zeta + \mathcal{L}_{\bm{u}_\zeta})\varrho_\zeta = 0$. Assume that $\nabla\cdot\bm{B} =  0$, or, equivalently, that $\mathcal{L}_{\bm{B}}\Omega = 0$. Using Cartan's magic formula, this condition can be rephrased as saying that $d\iota_{\bm{B}}\Omega = 0$. Using Cartan calculus and the definition of $\bm{u}_{\zeta}$, it follows that 
\[       
(\partial_{\zeta} + \mathcal{L}_{\bm{u}_{\zeta}})\varrho_{\zeta} = (\partial_{\zeta} + d\iota_{\bm{u}_{\zeta}} + \iota_{\bm{u}_{\zeta}}d)\rho_{\zeta} = \partial_{\zeta}I^*_{\zeta} \iota_{\bm{B}}\Omega -dI_\zeta^*(\iota_{\partial_\zeta}\iota_{\bm{B}}\Omega) .
\]
The expression on the right-hand-side vanishes because $\partial_\zeta (I^*_\zeta \lambda) = I^*_\zeta(\mathcal{L}_{\partial_{\zeta}}\lambda)$, for any differential form $\lambda$ on $D$. Now assume the converse: $(\partial_\zeta + \mathcal{L}_{\bm{u}_\zeta})\varrho_\zeta = 0$. The previous formulas imply
\begin{align*}
0=\partial_\zeta I^*_\zeta(\iota_{\bm{B}}\Omega)-I_\zeta^*(d\iota_{\partial_\zeta}\iota_{\bm{B}}\Omega) = I_\zeta^*(\iota_{\partial_\zeta}d\iota_{\bm{B}}\Omega) = I_\zeta^*(\nabla\cdot\bm{B}\,\iota_{\partial_\zeta}\Omega) = I_\zeta^*(\nabla\cdot \bm{B})\,\omega_\zeta.
\end{align*}
Since the preceding formula applies for all $\zeta$ we find $\nabla\cdot\bm{B} = 0$, as desired.

Next we show that force balance $(\nabla\times\bm{B})\times\bm{B} = \nabla p$ is equivalent to the pair of spatial dynamics equations \eqref{pressure_evo_ld} and \eqref{vel_evo_ld}. Assume force balance. In terms of differential forms, force balance is equivalent to $\iota_{\bm{B}} d\bm{B}^\flat = dp $, where the covering space metric $\widetilde{g}$ is used to lower indices. By Lemma \ref{B_param_lemma}, there is a vector field $\bm{u}$ on $\widetilde{Q}$ and a nowhere-vanishing function $n$ on $\widetilde{Q}$ such that $\bm{B} = n(\partial_\zeta + \bm{u})$, $I^*_\zeta n = \varrho_\zeta/\omega_\zeta$, and  
\begin{align*}
\mathcal{L}_{\bm{u}_\zeta}p_\zeta = I^*_\zeta(\mathcal{L}_{\bm{u}}p).
\end{align*}
From this it follows that 
\[
(\partial_{\zeta} + \mathcal{L}_{\bm{u}_\zeta}) p_{\zeta} = I^*_{\zeta} (\iota_{\partial_\zeta}d p+ \iota_{\bm{u}}dp) = I^*_{\zeta} (\iota_{\partial_\zeta}\iota_{\bm{B}} d\bm{B}^\flat +  \iota_{\bm{u}}\iota_{\bm{B}} d\bm{B}^\flat) = \frac{1}{n_\zeta}I^*_\zeta(\iota_{\bm{B}}\iota_{\bm{B}}d\bm{B}^\flat) = 0,
\]
where we have used $\iota_{\bm{B}}\iota_{\bm{B}} = 0$. A similar calculus gives us that 
\[
\frac{dp_{\zeta}}{n_{\zeta}} + d\bigg( n_{\zeta} (\left|\bm{u}_{\zeta} + \bm{N}_{\zeta} \right|^2  + N_{\zeta}^2)\bigg)  = \frac{1}{n_{\zeta}} I^*_{\zeta} (\iota_{\bm{B}}d\bm{B}^\flat) + d\bigg(n_\zeta\bigg[ 2 h_{\zeta}(\bm{N}_{\zeta}, \bm{u}_{\zeta}) + h_{\zeta}(\bm{u}_{\zeta}, \bm{u}_{\zeta})+I^*_{\zeta} \widetilde{g}(\partial_\zeta, \partial_\zeta) \bigg]\bigg)
\]
\[
=\frac{1}{n_{\zeta}} I^*_{\zeta}(\iota_{\bm{B}}d\bm{B}^\flat)+ d\left(n_\zeta I^*_{\zeta} \left[2\widetilde{g}\left(\frac{1}{n} \bm{B} - \partial_\zeta, \partial_\zeta\right) + \widetilde{g}\left(\frac{1}{n} \bm{B} - \partial_{\zeta},\frac{1}{n} \bm{B} - \partial_\zeta\right) +  \widetilde{g}\left(\partial_\zeta, \partial_\zeta\right) \right] \right)
\]
\[
=\frac{1}{n_{\zeta}} I^*_{\zeta}(\iota_{\bm{B}}d\bm{B}^\flat) + d\left(\frac{1}{n_{\zeta}}I^*_{\zeta}\widetilde{g}(\bm{B}, \bm{B})\right). 
\]
Compare this to 
\[
(\partial_{\zeta} + \mathcal{L}_{\bm{u}_\zeta})\left(n_{\zeta}(\bm{u}_{\zeta} + \bm{N}_{\zeta})^{\flat} \right)= (\partial_{\zeta} + \mathcal{L}_{\bm{u}_\zeta})n_{\zeta}I^*_{\zeta} \left( \iota_{\partial_{\zeta}} \widetilde{g} + \iota_{\bm{u} } \widetilde{g} \right) = (\partial_{\zeta} + d\iota_{\bm{u}_\zeta} + \iota_{\bm{u}_\zeta} d)I^*_{\zeta}\iota_{\bm{B}}\widetilde{g} 
\]
\[
= I^*_{\zeta}\left[ (\iota_{\partial_\zeta}d + d\iota_{\partial_\zeta} +  d\iota_{\bm{u}} + \iota_{\bm{u}} d) \bm{B}^\flat \right] = \frac{1}{n_{\zeta}} I^*_{\zeta}\iota_{\bm{B}} d\bm{B}^\flat + d\left( \frac{1}{n_\zeta} I^*_{\zeta} \widetilde{g}(\bm{B}, \bm{B}) \right).
\]
Hence we arrive at the desired equations \eqref{pressure_evo_ld} and \eqref{vel_evo_ld}. Now assume that Eqs.\,\eqref{pressure_evo_ld} and \eqref{vel_evo_ld} are satisfied. Similar manipulations to those above reveal 
\[
\frac{dp_{\zeta}}{n_{\zeta}} = (\partial_{\zeta} + \mathcal{L}_{\bm{u}_\zeta})\left(n_{\zeta}(\bm{u}_{\zeta} + \bm{N}_{\zeta})^{\flat} \right) -  d( n_{\zeta} (\left|\bm{u}_{\zeta} + \bm{N}_{\zeta} \right|^2 ) + N_{\zeta}^2) = \frac{1}{n_{\zeta}} I^*_{\zeta}\iota_{\bm{B}} d\bm{B}^\flat ,
\]
which implies that 
\begin{align}
I^*_{\zeta} dp = I^*_{\zeta}\iota_{\bm{B}} d\bm{B}^\flat .\label{perp_force_balance}
\end{align}
Furthermore, Eq.\,\eqref{pressure_evo_ld} gives us that 
\begin{align}
(\partial_{\zeta} + \mathcal{L}_{\bm{u}_\zeta}) p_{\zeta}  = I^*_{\zeta} \left( \iota_{\partial_\zeta} dp + \iota_{\bm{u}} dp\right) = \frac{1}{n_\zeta} I^*_{\zeta}\iota_{\bm{B}} dp = 0.\label{parallel_force_balance}
\end{align}
It is straightforward to confirm that Eqs.\,\eqref{perp_force_balance} and \eqref{parallel_force_balance}, which must be satisfied for all $\zeta$, imply force balance, $\iota_{\bm{B}}d\bm{B}^\flat = dp$, as desired.

To complete the proof, we now must show that the boundary condition $\bm{B}\cdot\bm{n} = 0$ on $\partial{\widetilde{Q}}$ is equivalent to the tangential boundary condition on $\bm{u}_\zeta$. By the definition of toroidal coordinates, the boundary $\partial\widetilde{Q}$ is given as the level set $\{x^2 + y^2 = 1\}$. The normal vector $\bm{n}$ is therefore proportional to $\bm{n}_0 = x\,\nabla x + y\nabla y$, where gradients are taken using $\widetilde{g}$. The normal component of $\bm{B}$ is therefore given by
\begin{align*}
\bm{B}\cdot \bm{n} = \bm{B}\cdot \frac{\bm{n}_0}{|\bm{n}_0|} = n\,(\partial_\zeta + \bm{u})\cdot \frac{\bm{n}_0}{|\bm{n}_0|} = n\,\bm{u}\cdot \frac{\bm{n}_0}{|\bm{n}_0|} = n\,\frac{x\,u^x_\zeta + y\,u^y_\zeta}{|\bm{n}_0|}.
\end{align*}
So vanishing of $\bm{B}\cdot\bm{n}$ along $\partial\widetilde{Q}$ is equivalent to vanishing of $x\,u^x_\zeta + y\,u^y_\zeta$ along $\partial D$. But vanishing of $x\,u^x_\zeta + y\,u^y_\zeta$ along $\partial D$ is equivalent to tangency of $\bm{u}_\zeta$ to the boundary of the unit disc $\partial D$.
\end{proof}

Theorem \ref{sdf} provides the spatial dynamics formulation for MHS on covering space $\widetilde{Q}$. The remainder of this Article begins the process of applying ideas from dynamical systems theory to these equations in order to shed new light on three-dimensional MHD equilibria. But before proceeding it is important to emphasize the connection between the spatial dynamics formulation on $\widetilde{Q}$ and the original boundary value problem on $Q$. Otherwise the interpretation of solutions of the spatial dynamics equations will be ambiguous.

In principle, Lemma \ref{deck_transformations_lemma} already provides the precise link between solutions on $\widetilde{Q}$ and solutions on $Q$. However, using the machinery developed so far, that connection may be described in more concrete terms. According to Lemma \ref{deck_transformations_lemma}, MHS solutions on $Q$ are equivalent to MHS solutions on $\widetilde{Q}$ that are invariant under "deck transformations". In toroidal coordinates $(\zeta,x,y)$ on $\widetilde{Q}$ a deck transformation is simply a mapping of the form
\begin{align*}
(\zeta,x,y)\mapsto (\zeta + 2\pi k,x,y),\quad k\in \mathbb{Z}.
\end{align*}
If $(\rho_\zeta,\bm{u}_\zeta,p_\zeta)$ denotes a solution of the spatial dynamics equations on $\widetilde{Q}$ its deck transformation is $(\rho^\prime_\zeta,\bm{u}^\prime_\zeta,p^\prime_\zeta)$, where
\begin{align*}
\rho^\prime_\zeta = \rho_{\zeta +2\pi k}, \quad \bm{u}^\prime_\zeta = \bm{u}_{\zeta + 2\pi k},\quad p^\prime_\zeta = p_{\zeta + 2\pi k}.
\end{align*}
This shows that a solution of the spatial dynamics equations on $\widetilde{Q}$ will be invariant under deck transformations if and only if it is $2\pi$-periodic in time $\zeta$. In other words, the only solutions of the spatial dynamics equations that correspond to physical equilibria back in $Q$ are the periodic orbits with period $2\pi$. But there is no loss of generality here because every MHS solution on $Q$ corresponds to a unique periodic orbit of the spatial dynamics equations. This simple reinterpretation of Lemma \ref{deck_transformations_lemma} requires that the toroidal coordinates $\widetilde{\varphi} = (\zeta,x,y)$ on $\widetilde{Q}$ be related to a toroidal coordinate system $\varphi$ on $Q$ by the following commutative diagram.
\[\begin{tikzcd}
	{\widetilde{Q}} & {} & Q \\
	\\
	{\mathbb{R}\times D} && {\mathbb{S}^1\times D}
	\arrow["{\mathcal{P}}", from=1-1, to=1-3]
	\arrow[from=3-1, to=3-3]
	\arrow["{\widetilde{\varphi}}"{description}, from=1-1, to=3-1]
	\arrow["\varphi"{description}, from=1-3, to=3-3]
\end{tikzcd}\]
The bottom horizontal arrow denotes the mapping $(\zeta,x,y)\mapsto (\zeta\text{ mod }2\pi,x,y)$.

\section{MHS spatial dynamics are variational\label{variational_formulation}}
Variational principles for MHS go back at least as far as the work of Kruskal and Kulsrud\cite{Kruskal_Kulsrud_1958}, who showed that MHS solutions arise as constrained extrema of plasma potential energy. The constraints reflect kinematic features of dynamical ideal MHD, including advection of magnetic flux, mass, and entropy. Newcomb\cite{Newcomb_1962} discovered an elegant method for incorporating the Kruskal-Kulsrud constraints using configuration maps familiar from Lagrangian hydrodynamics. These ideas lead naturally to a variational principle for the spatial dynamics formulation of MHS.

With sufficient care, a spatial dynamics variational principle emerges from a deductive chain that starts with an MHS variational principle implicit in Newcomb's work\cite{Newcomb_1962}. We will not provide the details of this argument here. Instead, we will present the end-result and confirm its correctness through direct calculation. But, before precisely formulating the variational principle, we need to introduce a label map for magnetic field lines.

Geometrically, a magnetic field line is a $1$-dimensional curve that is everywhere tangent to the magnetic field $\bm{B}$. When $\bm{B}$ is toroidal, and toroidal coordinates $(\zeta,x,y)$ have been introduced, the field line takes on a dynamical character. The line's intersection with hypersurface $\zeta$ determines a unique "spatial" point $z(\zeta) = (x(\zeta),y(\zeta))\in D$ in the parameterizing disc; as $\zeta$ changes this spatial point $z(\zeta)$ moves. Because $\bm{B} = n(\partial_\zeta + \bm{u})$, the velocity of $z(\zeta)$ is precisely $\partial_\zeta z = \bm{u}_\zeta(z(\zeta))$. This leads to a picture of MHS solutions as comprised of many particles $z$ in the plane that flow along stream lines of the velocity field $\bm{u}_\zeta$. In the parlance of relativity theory, the world line of each particle is a field line.

Bookkeeping of the numerous particles $z\in D$ requires introducing a label map $\psi_\zeta$. The map $\psi_\zeta$ should assign a unique label $Z = (U,V) = \psi_\zeta(z)$ to each particle $z$ that does not change as $z$ evolves through time. Time-independence of the label attached to $z$ implies the velocity of the configuration map $\partial_\zeta\psi_\zeta$ determines the particle velocity $\bm{u}_\zeta$ according to
\begin{align*}
\bm{u}_\zeta(x,y) = -[D\psi_\zeta(x,y)]^{-1}\partial_\zeta\psi_\zeta(x,y).
\end{align*}
This formula for $\bm{u}_\zeta$ and the spatial dynamics equations \eqref{pressure_evo_ld}--\eqref{flux_evo_ld} further imply that hypersurface pressure $p_\zeta$ and toroidal flux $\varrho_\zeta$ both freeze when expressed in label space. In terms of differential forms, this means $\partial_\zeta (\psi_{\zeta\,*}p_\zeta) = 0$ and $\partial_\zeta(\psi_{\zeta\,*}\varrho_\zeta )= 0$. If the frozen label-space expressions for $p_\zeta$ and $\varrho_\zeta$ are denoted $\hat{p}$ and $\hat{\varrho}$ then $p_\zeta$ and $\varrho_\zeta$ are completely determined by $\psi_\zeta$ according to
\begin{align*}
p_\zeta = \psi_\zeta^*\hat{p},\quad \varrho_\zeta = \psi_\zeta^*\hat{\varrho}.
\end{align*}

Without loss of generality, we may assume that the space of labels $Z=(U,V)$ where $\psi_\zeta$ takes its values is $D(\Gamma_0)$, the centered disc in $\mathbb{R}^2$ with area given by the total toroidal flux $\Gamma_0 = \int_D \varrho_\zeta$. Furthermore, the frozen label space toroidal flux may always be assumed to be in the standard form $\hat{\varrho} = dU\wedge dV$. These statements may be understood as consequences of global existence of Clebsch-like coordinates\cite{Haeseleer_1991,XJ_2006} in covering space $\widetilde{Q}$. The same simplifications do not apply back in $Q$ due to field lines intersecting hypersurfaces multiple times.

The discussion so far leads to the following action functional for MHS spatial dynamics.

\begin{definition}
Let $D$ be the standard unit disc in $\mathbb{R}^2$ and let $D(\Gamma_0)$ denote the disc in $\mathbb{R}^2$ with area $\Gamma_0$. Let $\bm{N}_\zeta,N_\zeta,\omega_\zeta,h_\zeta$ be a shift vector, lapse function, hypersurface area element, and hypersurface metric on $D$ associated with some toroidal coordinates on covering space $\widetilde{Q}$. Fix $\Gamma_0 > 0$, a smooth function $\hat{p}:D(\Gamma_0)\rightarrow \mathbb{R}$, and $\zeta_1,\zeta_2\in\mathbb{R}$ with $\zeta_1 \leq \zeta_2$.   Finally, introduce the space $\mathcal{D} = \text{Diff}(D,D(\Gamma_0))$ of diffeomorphisms from $D$ to $D(\Gamma_0)$.
The \textbf{spatial dynamics action} with toroidal flux $\Gamma_0$ on the interval $[\zeta_1,\zeta_2]$ is the real-valued function $\mathcal{S}_{\hat{p}}$ of paths $\psi:[\zeta_1,\zeta_2]\rightarrow \mathcal{D}:\zeta\mapsto \psi_\zeta$ given as follows.
\begin{align}
\mathcal{S}_{\hat{p}}(\psi) &= \int_{\zeta_1}^{\zeta_2}\int_D\bigg[\frac{1}{2}\left(\frac{\varrho_\zeta}{\omega_\zeta}\right)^2|\bm{u}_\zeta+ \bm{N}_\zeta|^2 + \frac{1}{2}\left(\frac{\varrho_\zeta}{\omega_\zeta}\right)^2\,N_\zeta^2 - p_\zeta \bigg]\omega_\zeta\,d\zeta,\label{spatial_dynamics_action}
\end{align}
where
\begin{align*}
\bm{u}_\zeta= -[D\psi_\zeta]^{-1}\partial_\zeta\psi_\zeta,\quad p_\zeta = \psi_\zeta^*\hat{p},\quad \varrho_\zeta = \psi_\zeta^*(dU\wedge dV).
\end{align*}
Here the $2$-norm $|\cdot|$ is taken with respect to the hypersurface metric $h_\zeta$.
A path $\psi$ is a \textbf{critical point} for the spatial dynamics action with toroidal flux $\Gamma_0$ if
\begin{align*}
\delta \mathcal{S}_{\hat{p}}(\psi) = 0,\quad \delta \psi_\zeta=0, \text{ for }\zeta = \zeta_1\text{ or }\zeta = \zeta_2.
\end{align*}
\end{definition}

The integrand for the spatial dynamics action is equivalent to $\tfrac{1}{2}|\bm{B}|^2 - p$, which appears in the work of Kruskal and Kulsrud, Eq.\,(E1) with $\gamma=0$, Grad\cite{Grad_1964}, and Bhattacharjee\cite{Bhattacharjee_1984}, among others. Starting from Newcomb's equilibrium variational principle, this difference between magnetic and internal energy arises as follows. In Newcomb's approach to equilibria, the magnetic field $\bm{B}$ is given as a volume-preserving rearrangement of a reference magnetic field $\bm{B}_0$. This reference specifies the topological properties of field lines; restricting $\bm{B}_0$ to the straight field-line form\cite{Greene_1962} connects Newcomb's theory with Bhattacharjee's. Equilibria then arise as extrema for magnetic energy, subject to these constraints. Note that magnetic energy is not the same as plasma potential energy! The volume-preserving constraint on the rearrangment relating $\bm{B}$ with $\bm{B}_0$ may be lifted by introducing a Lagrange multiplier. When this is done, the Lagrange multiplier assumes the physical interpretation of pressure, and the value of the action functional changes. Finally, the spatial dynamics action emerges by first passing to covering space, introducing the parameterization of $(\bm{B},p)$ using $\varrho_\zeta,\bm{u}_\zeta,p_\zeta$, and then dropping a temporal boundary term from the action functional. It is interesting that $\tfrac{1}{2}|\bm{B}|^2 - p$ appears in this roundabout fashion, rather than from the plasma potential energy with $\gamma=0$.

Applying Hamilton's principle to the spatial dynamics action leads to a variational characterization of MHS spatial dynamics. The variational principle itself is described in Theorem \ref{sd_vp}; Lemma \ref{ele_lemma} provides a convenient independent characterization of the Euler-Lagrange equations associated with the spatial dynamics action $\mathcal{S}_{\hat{p}}$.

\begin{lemma}\label{ele_lemma}
The Euler-Lagrange equations associated with the spatial dynamics action are
\begin{gather}
(\partial_\zeta + \mathcal{L}_{\bm{u}_\zeta})\bigg(n_\zeta(\bm{u}_{\zeta} + \bm{N}_\zeta)^\flat\bigg)   =\frac{dp_\zeta}{n_\zeta} +  d\bigg(n_\zeta\,(|\bm{u}_\zeta + \bm{N}_\zeta|^2 + N_\zeta^2) \bigg).\label{ele}
\end{gather}
The Eulerian velocity $\bm{u}_\zeta$ must be tangent to $\partial D$ for each $\zeta$. The toroidal flux and hypersurface pressure must obey the evolution equations $\partial_\zeta\varrho_\zeta +\mathcal{L}_{\bm{u}_\zeta}\varrho_\zeta = 0$ and $\partial_\zeta p_\zeta +\mathcal{L}_{\bm{u}_\zeta}p_\zeta = 0$.
\end{lemma}
\begin{proof}
First it is useful to record formulas for the first variations of various terms in the hydrodynamic MHS action. Let $\bm{\xi}_\zeta = -(T\psi_\zeta)^{-1}\delta\psi_\zeta$ denote the Eulerianized first variation of $\psi_\zeta$. We have
\begin{align*}
\delta\bm{u}_\zeta = \partial_\zeta\bm{\xi}_\zeta  + \mathcal{L}_{\bm{u}_\zeta}\bm{\xi}_\zeta,\quad \delta \varrho_\zeta  = -\mathcal{L}_{\bm{\xi}_\zeta}\varrho_\zeta,\quad \delta p_\zeta = -\mathcal{L}_{\bm{\xi}_\zeta}p_\zeta.
\end{align*}
Now the first variation of the spatial dynamics action may be computed directly as follows:
\begin{align}
\delta\mathcal{S}_{\hat{p}} & = \int_{\zeta_1}^{\zeta_2}\int_D \left(\frac{\varrho_\zeta}{\omega_\zeta}\right)\left(|\bm{u}_\zeta + \bm{N}_\zeta|^2 + N_\zeta^2\right)\,\delta\varrho_\zeta\,d\zeta+\int_{\zeta_1}^{\zeta_2}\int_D \left(\frac{\varrho_\zeta}{\omega_\zeta}\right)^2\,(\bm{u}_\zeta + \bm{N}_\zeta)^\flat(\delta \bm{u}_\zeta)\,\omega_\zeta\,d\zeta-\int_{\zeta_1}^{\zeta_2}\int_D \delta p_\zeta\,\omega_\zeta\,d\zeta\nonumber\\
&= \int_{\zeta_1}^{\zeta_2}\int_D d\left[\left(\frac{\varrho_\zeta}{\omega_\zeta}\right)(|\bm{u}_\zeta + \bm{N}_\zeta|^2 + N_\zeta^2)\right](\bm{\xi}_\zeta)\,\varrho_\zeta + \int_{\zeta_1}^{\zeta_2}\int_D \left(\frac{\omega_\zeta}{\varrho_\zeta}\right)\,dp_\zeta(\bm{\xi}_\zeta)\,\varrho_\zeta\nonumber\\
&-\int_{\zeta_1}^{\zeta_2}\int_D(\partial_\zeta + \mathcal{L}_{\bm{u}_\zeta})\left[\left(\frac{\varrho_\zeta}{\omega_\zeta}\right)\,(\bm{u}_\zeta + \bm{N}_\zeta)^\flat\right](\bm{\xi}_\zeta)\,\varrho_\zeta + \int_D \left(\frac{\varrho_\zeta}{\omega_\zeta}\right)(\bm{u}_\zeta + \bm{N}_\zeta)^\flat(\bm{\xi}_\zeta)\,\varrho_\zeta\bigg|_{\zeta_1}^{\zeta_2}.
\end{align}
Obtaining the second equality requires some spatial integration by parts that generates spatial boundary terms. These boundary terms all vanish because $\bm{u}_\zeta$ is tangent to $\partial D$. Now suppose that $\psi$ is a critical point. The last term vanishes because $\delta \psi_{\zeta_1} = \delta\psi_{\zeta_2} = 0$. The remaining terms must vanish for each $\bm{\xi}_\zeta$. By the nowhere-vanishing property of the toroidal flux $\varrho_\zeta$, we conclude that Eq.\,\eqref{ele} provides the Euler-Lagrange equations, as claimed. 

The velocity $\bm{u}_\zeta$ must be tangent to $\partial D$ because the diffeomorphism $\psi_\zeta:D\rightarrow D(\Gamma_0)$ maps $\partial D$ diffeomorphically onto $\partial D(\Gamma_0)$ for each $\zeta$. The evolution equations $\partial_\zeta\varrho_\zeta + \mathcal{L}_{\bm{u}_\zeta}\varrho_\zeta = 0$ and $\partial_\zeta p_\zeta + \mathcal{L}_{\bm{u}_\zeta}p_\zeta = 0$ follow from time differentiation of the definitions of $\varrho_\zeta$ and $p_\zeta$ in terms of $\psi_\zeta$: $\varrho_\zeta = \psi_\zeta^*(dU\wedge dV)$ and $p_\zeta = \psi_\zeta^*\hat{p}$.

\end{proof}

\begin{theorem}\label{sd_vp}
Let $\varrho_\zeta$ be an evolving nowhere-vanishing $2$-form on $D$, $\bm{u}_\zeta$ an evolving vector field on $D$, and $p_\zeta$ and evolving function on $D$. The following are equivalent.
\begin{itemize}
\item $\varrho_\zeta$, $\bm{u}_\zeta$, and $p_\zeta$ solve the spatial dynamics equations \eqref{pressure_evo_ld}-\eqref{vel_evo_ld}, with $\bm{u}_\zeta$ obeying tangential boundary conditions.
\item There is a function $\hat{p}:D(\Gamma_0)\rightarrow\mathbb{R}$ and a critical point $\psi$ of $\mathcal{S}_{\hat{p}}$ such that 
\begin{align*}
\bm{u}_\zeta = -[D\psi_\zeta]^{-1}\partial_\zeta\psi_\zeta,\quad \varrho_\zeta = \psi_\zeta^*(dU\wedge dV),\quad p_\zeta = \psi_\zeta^*\hat{p}.
\end{align*}
\end{itemize}
\end{theorem}

\begin{proof}
That the second bullet implies the first bullet is equivalent to Lemma \ref{ele_lemma}. So we only need to prove that the first bullet implies the second.

Assume that $\varrho_\zeta$, $\bm{u}_\zeta$, $p_\zeta$ solve the MHS spatial dynamics equations \eqref{pressure_evo_ld}-\eqref{vel_evo_ld}, with $\bm{u}_\zeta$ satisfying tangential boundary conditions. We must construct a critical point $\psi$ for the spatial dynamics action that recovers this solution in the appropriate sense. By tangency of $\bm{u}_\zeta$ to $\partial D$, the time-dependent flow map $\Phi_{\zeta_b,\zeta_a}:z(\zeta_a) \mapsto z(\zeta_b)$ for $\bm{u}_\zeta$ defines a $2$-parameter family of diffemorphisms $\Phi_{\zeta_b,\zeta_a}:D\rightarrow D$. Fix $\zeta_1\in\mathbb{R}$ and let $\Psi_\zeta = (\Phi_{\zeta,\zeta_1})^{-1}$. By the relationship between flows and Lie derivatives\cite{Abraham_2008}, and the spatial dynamics equation \eqref{flux_evo_ld}, the $2$-form $\Psi_{\zeta\,*}\varrho_\zeta$ on $D$ is time-independent:
\begin{align*}
\partial_\zeta(\Psi_{\zeta\,*}\varrho_\zeta) = \partial_{\zeta}(\Phi_{\zeta,\zeta_1}^*\varrho_\zeta) = \Phi_{\zeta,\zeta_1}^*(\partial_\zeta\varrho_\zeta + \mathcal{L}_{\bm{u}_\zeta}\varrho_\zeta) = 0.
\end{align*}
In light of the spatial dynamics equation \eqref{pressure_evo_ld}, the same calculation shows that the function $\Psi_{\zeta\,*}p_\zeta$ on $D$ is time independent.

Let $\hat{R} =\Psi_{\zeta\,*}\varrho_\zeta$ and $\hat{P} =\Psi_{\zeta\,*}p_\zeta $. Introduce the total toroidal flux $\Gamma_0 = \int_D\hat{R} = \int_D\varrho_\zeta$. Since $\Gamma_0 = \int_{D(\Gamma_0)}dU\wedge dV$, a result due to Bayaga\cite{Banyaga_1974} implies there is a diffeomorphism $\eta:D\rightarrow D(\Gamma_0)$ such that $\eta_*\hat{R} = dU\wedge dV$. The time-dependent diffeomorphism $\psi_\zeta = \eta\circ \Psi_\zeta : D\rightarrow D(\Gamma_0)$ satisfies
\begin{align*}
\psi_\zeta^*(dU\wedge dV) = \Psi_\zeta^*\eta^*(dU\wedge dV)=\Psi_\zeta^*\hat{R} = \varrho_\zeta.
\end{align*}
If $\hat{p}= \eta_*\hat{P}$, it also satisfies
\begin{align*}
\psi_\zeta^*\hat{p} = \Psi_\zeta^*\eta^*\hat{p} = \Psi_\zeta^*\hat{P}=p_\zeta.
\end{align*}
Finally, from the time derivative of $\psi_\zeta\circ \Phi_{\zeta,\zeta_1} = \eta$ we obtain, for each $z\in D$,
\begin{align*}
0 = \partial_\zeta\psi_\zeta(\Phi_{\zeta,\zeta_1}(z)) + D\psi_\zeta(\Phi_{\zeta,\zeta_1}(z))\,\partial_\zeta\Phi_{\zeta,\zeta_1}(z) =  \partial_\zeta\psi_\zeta(\Phi_{\zeta,\zeta_1}(z)) +D\psi_\zeta(\Phi_{\zeta,\zeta_1}(z))\,\bm{u}_\zeta(\Phi_{\zeta,\zeta_1}(z)),
\end{align*}
which is equivalent to $\bm{u}_\zeta = -[D\psi_\zeta]^{-1}\partial_\zeta\psi_\zeta$. Since $\bm{u}_\zeta$ satisfies the spatial dynamics equation \eqref{vel_evo_ld}, we conclude that $\psi:\zeta\mapsto \psi_\zeta$ is a critical point for $\mathcal{S}_{\hat{p}}$.

\end{proof}

\section{Conservation laws via Noether's theorem\label{cons_noether}}
Section \ref{variational_formulation} showed that the spatial dynamics formulation of MHS derived in Section \ref{sec:direct} comprises a variational dynamical system. Noether's theorem therefore guarantees that every symmetry of the spatial dynamics action \eqref{spatial_dynamics_action} corresponds to a conservation law for the spatial dynamics equations \eqref{pressure_evo_ld}-\eqref{vel_evo_ld}. Two of these conservation laws warrant special attention.

Time translation invariance of $\mathcal{S}_{\hat{p}}$ implies conservation of "energy." Since $\zeta$ is not physical time, the physical interpretation of this "energy" is initially unclear. 
Suppose that all geometric quantities associated with the hypersurface foliation, namely $h_\zeta$, $\bm{N}_\zeta$, $N_\zeta$, and $\omega_\zeta$, are $\zeta$-independent. This can always be arranged when the spatial domain $Q$ is a volume of revolution. Then the spatial dynamics action is time-independent. Straightforward application of Noether's theorem implies that the associated conserved quantity is
\begin{align}
E(\bm{u},\varrho,p) = \int_D\left[\frac{1}{2}\left(\frac{\varrho}{\omega}\right)^2|\bm{u}|^2-\frac{1}{2}\left(\frac{\varrho}{\omega}\right)^2(|\bm{N}|^2 + N^2) + p\right]\omega.\label{energy_invariant}
\end{align}
Proving conservation of $E$ directly using the spatial dynamics equations \eqref{pressure_evo_ld}-\eqref{vel_evo_ld} presents an instructive exercise for interested readers. The three terms appearing here admit simple interpretations. By Lemma \ref{B_param_lemma}, the magnetic field may be written as the sum of a toroidal component $\bm{B}_T = (\varrho/\omega)\partial_\zeta$ and a poloidal component $\bm{B}_P = (\varrho/\omega)\bm{u}$. The first term in \eqref{energy_invariant} gives the hypersurface poloidal magnetic field energy. The second term gives (minus) the hypersurface toroidal magnetic field energy. The third term is the hypersurface internal energy. The difference between toroidal and poloidal magnetic field energies also appears in the variational principle for the Grad-Shafranov equation\cite{RWhite_2014} and its generalizations\cite{BKM_2020a,BKM_2020b}. Sign indefiniteness of $E$ foreshadows the discussions in Sections \ref{sec:fs_formulation} and \ref{sec:snakes}, where the well-known mixed elliptic-hyperbolic type of the MHS equations manifests itself as a normally-hyperbolic slow manifold\cite{Fenichel_1979,Lorenz_1986,Lorenz_1987,Lorenz_1992,Gorban_2004,MacKay_2004,Gorban_2018,Burby_Klotz_2020} for the spatial dynamics equations.

Rotation invariance of $\mathcal{S}_{\hat{p}}$ implies conservation of "angular momentum." We will not discuss this conservation law further because toroidal confinement devices generally do not admit poloidal rotation symmetries; toroidicity spoils poloidal rotation invariance. However, it would be interesting to determine the Noether conserved quantity associated with "quasisymmetry"\cite{Boozer_1983,NZ_1988,Burby_Qin_2013,Helander_2014,BKM_2020b,Rodriguez_2020}. While the asymptotic analysis of Garren and Boozer\cite{Garren_Boozer_1991} suggests that three-dimensional quasisymmetric MHS solutions cannot exist, more recent numerical studies by Landreman and Paul\cite{Landreman_Paul_2022} indicate that deviations from quasisymmetry can be extremely small. This, together with Noether's theorem, suggests the presence of an adiabatic invariant\cite{Kruskal_1962,Burby_Squire_2020} associated with quasisymmetry.

Invariance of $\mathcal{S}_{\hat{p}}$ under relabeling of particles implies conservation of "circulation".  Examples from non-dissipative continuum mechanics include the Vlasov-Poisson system, where the conserved circulation corresponds to the Poincar\'e integral invariant $\oint p\,dq$, and the ideal barotropic Euler equations, where the conserved circulation corresponds to ordinary fluid circulation. In the fluid context more generally, Cotter and Holm\cite{Cotter_Holm_2012} provide a detailed study the Noether conserved quantities associated with particle relabeling symmetry in the presence of advected quantities. The conserved "ciculation" for MHS spatial dynamics warrants careful consideration. 

In contrast to time translation symmetry, which invariably breaks in non-symmetric toroidal domains $Q$, the spatial dynamics action $\mathcal{S}_{\hat{p}}$ always enjoys invariance under a group of particle relabeling transformations. The group $\mathcal{G}$ comprises diffeomorphisms of label space $\eta:D(\Gamma_0)\rightarrow D(\Gamma_0)$ that preserve both area, $\eta^*(dU\wedge dV) = dU\wedge dV$, and reference pressure, $\eta^*\hat{p} = \hat{p}$. The transformation of a path $\psi$ by a relabeling transformation $\eta\in\mathcal{G}$ is $\psi_{\eta}:\zeta\mapsto \eta\circ \psi_\zeta$. The spatial dynamics action $\mathcal{S}_{\hat{p}}$ satisfies $\mathcal{S}_{\hat{p}}(\psi_\eta) = \mathcal{S}_{\hat{p}}(\psi)$ because relabeling does not change toroidal flux or hypersurface pressure,
\begin{align*}
(\eta\circ \psi_\zeta)^*(dU\wedge dV) = \psi_\zeta^*\eta^*(dU\wedge dV) = \psi_\zeta^*(dU\wedge dV),\quad (\eta\circ \psi_\zeta)^*\hat{p}  = \psi_\zeta^*\eta^*\hat{p} = \psi_\zeta^*\hat{p},
\end{align*}
nor does it change field line velocity,
\begin{align*}
-[D(\eta\circ\psi_\zeta)]^{-1}\partial_\zeta(\eta\circ\psi_\zeta) = -[D\psi_\zeta]^{-1}[D\eta]^{-1}(D\eta\,\partial_\zeta\psi_\zeta) = -[D\psi_\zeta]^{-1}\,\partial_\zeta\psi_\zeta.
\end{align*}
This explains why Noether's theorem must imply a conservation law related to the group $\mathcal{G}$.

Theorem \ref{relabeling_symmetry} below finds this conservation law. Its proof uses the basic argument underlying Noether's theorem to show that the conserved circulation associated with $\mathcal{G}$ corresponds circulation of $\pi_\zeta$ (c.f. \eqref{vel_evo_ld} and \eqref{vel_evo_ind}) around contours of constant hypersurface pressure $p_\zeta$. An equivalent result can be found without recourse to Noether's theorem as follows. Let $(\bm{B},p)$ denote an MHS solution on $\widetilde{Q}$ and suppose that $P_0$ is a regular value for $p$ so that the level set $\Sigma(P_0) = \{p = P_0\}$ is an embedded $2$-torus. The intersection of $\Sigma(P_0)$ with each disc $D_\zeta$ in the hypersurface foliation defines a simple closed curve $C_{\zeta}(P_0)\subset \widetilde{Q}$ that bounds a disc $D_{\zeta}(P_0)\subset D_\zeta$. The current that passes through $D_\zeta(P_0)$ is $I_\zeta(P_0) = \int_{D_\zeta(P_0)}\bm{J}\cdot d\bm{S}$, where $\bm{J} = \nabla\times \bm{B}$ denotes the current density. Finally, since $\bm{J}\cdot\nabla p = 0$ and $\nabla\cdot\bm{J} = 0$, the current passing through $D_{\zeta_1}(P_0)$ must equal the current passing through $D_{\zeta_2}(P_0)$, which is equivalent to Theorem \ref{relabeling_symmetry}. 

\begin{definition}
A reference pressure $\hat{p}:D(\Gamma_0)\rightarrow\mathbb{R}$ is \textbf{nondegenerate} if $\hat{p}$ is constant on $\partial D(\Gamma_0)$ and $d\hat{p}$ is nonzero everywhere except at a single point $a_0\in D(\Gamma_0)$. The point $a_0$ will be referred to as the \textbf{reference magnetic axis}.
\end{definition}

\begin{theorem}\label{relabeling_symmetry}
Assume that the reference pressure in the spatial dynamics action $\mathcal{S}_{\hat{p}}$ is nondegenerate. Let $P_0$ denote any regular value for $\hat{p}$ and let $c_\zeta(P_0) = \{(x,y)\mid p_\zeta(x,y) = P_0\}$ denote the isobar for $p_\zeta=\psi_\zeta^*\hat{p}$ with pressure $P_0$. Along any solution $\psi_\zeta$ of the Euler-Lagrange equations for $\mathcal{S}_{\hat{p}}$, and for any pair of times $\zeta_1\leq \zeta_2$, circulation around the isobar is conserved:
\begin{align}
\oint_{c_{\zeta_1}(P_0)}\pi_{\zeta_1} = \oint_{c_{\zeta_2}(P_0)}\pi_{\zeta_2}.
\end{align}
\end{theorem}

\begin{proof}
The Lie algebra $\mathfrak{g}$ for $\mathcal{G}$ contains vector fields $\Xi$ on $D(\Gamma_0)$ such that $\mathcal{L}_{\Xi}(dU\wedge dV) = 0$ and $\mathcal{L}_{\Xi}\hat{p} = 0$. The first condition implies there is a function $h$ on $D(\Gamma_0)$ such that $\iota_{\Xi}(dU\wedge dV) = dh$. The second condition implies that $dh\wedge d\hat{p} = 0$. By non-degeneracy of $\hat{p}$, the last condition implies there is a smooth single-variable function $H:\mathbb{R}\rightarrow\mathbb{R}$ such that $h = H\circ \hat{p}$. Overall we have shown that $\Xi$ is in $\mathfrak{g}$ if and only if there is a smooth single-variable function $H$ such that
\begin{align*}
\Xi = \partial_VH(\hat{p})\,\partial_U - \partial_UH(\hat{p})\,\partial_V.
\end{align*}

As shown previously, if $\eta\in \mathcal{G}$ is any relabeling transformation and $\psi$ is any critical point for $\mathcal{S}_{\hat{p}}$ then $\mathcal{S}_{\hat{p}}(\psi_\eta) = \mathcal{S}_{\hat{p}}(\psi)$. Differentiating this formula in $\eta$ along the direction $\Xi\in\mathfrak{g}$ and using the fact that $\psi$ satisfies the Euler-Lagrange equations implies
\begin{align}
0 = \int_D \pi_\zeta(\bm{\xi}_\zeta)\,\varrho_\zeta\bigg|_{\zeta_1}^{\zeta_2},\quad \bm{\xi}_\zeta = -\psi_\zeta^*\Xi.\label{proto_noether}
\end{align}
The integral $I_\zeta = \int_D \pi_\zeta(\bm{\xi}_\zeta)\,\varrho_\zeta$ simplifies according to
\begin{align*}
I_\zeta = \int_D\pi_\zeta\wedge (\iota_{\bm{\xi}_\zeta}\varrho_\zeta) = -\int_D\pi_\zeta\wedge \psi_\zeta^*(\iota_{\Xi}[dU\wedge dV]) = -\int_D\pi_\zeta\wedge \psi_\zeta^*dh = \int_{\partial D}(H\circ p_\zeta) \pi_\zeta - \int_D (H\circ p_\zeta)\,d\pi_\zeta.
\end{align*}

Now consider the level set $\gamma = \{\hat{p} = P_0\}$ in $D(\Gamma_0)$. The reference magnetic axis cannot be contained in $\gamma$ because $P_0$ is a regular value for $\hat{p}$. On the interior disc $D_\gamma(\Gamma_0)$ spanned by $\gamma$ the reference pressure assumes a minimum value $p^{\text{in}}_{\text{min}}$ and a maximum value $p^{\text{in}}_{\text{max}}$. On the complement $D(\Gamma_0) - D_\gamma(\Gamma_0)$, the reference pressure assumes a minimum value $p^{\text{out}}_{\text{min}}$ and a maximum value $p^{\text{out}}_{\text{max}}$. By nondegeneracy of $\hat{p}$, the open intervals $(p^{\text{in}}_{\text{min}},p^{\text{in}}_{\text{max}})$ and $(p^{\text{out}}_{\text{min}},p^{\text{out}}_{\text{max}})$ must be disjoint. If $H$ is the indicator function for the set $(p^{\text{in}}_{\text{min}},p^{\text{in}}_{\text{max}})\subset\mathbb{R}$ then $H\circ p_\zeta$ is the indicator function for the interior disc $D_\zeta(P_0)$ spanned by the level set $c_\zeta(P_0) = \{(x,y)\mid p_\zeta(x,y) = P_0\}$. The equation \eqref{proto_noether} therefore implies
\begin{align*}
0 = \int_{\partial D}(H\circ p_\zeta) \pi_\zeta\bigg|_{\zeta_1}^{\zeta_2} - \int_D (H\circ p_\zeta)\,d\pi_\zeta\bigg|_{\zeta_1}^{\zeta_2} = -\int_{D_\zeta(P_0)}d\pi_\zeta\bigg|_{\zeta_1}^{\zeta_2} = -\int_{c_\zeta(P_0)}\pi_\zeta\bigg|_{\zeta_1}^{\zeta_2},
\end{align*}
where the last equality follows from Stokes theorem.
\end{proof}

Conservation of circulation around an isobar for MHS spatial dynamics fits into a more general pattern that applies to all variational dynamical systems of Euler-Poincar\'e type\cite{Holm_EP_1998}. In fact, Theorem \ref{relabeling_symmetry} follows directly from the so-called Klevin-Noether theorem, first formulated by Holm, Marsden, and Ratiu\cite{Holm_EP_1998}, Theorem 6.2. That theorem equates the rate of change of circulation $I$ with the closed loop integral of the $1$-form density $\varrho^{-1}\delta\ell/\delta a\diamond a$. Here $\ell$ denotes the Lagrangian and $a$ denotes an advected parameter. For MHS spatial dynamics, the loop integral vanishes when the loop is an isobar.


\section{MHS spatial dynamics as a Lie-Poisson system\label{sec:lie_poisson}}
The variational formulation of MHS spatial dynamics from Section \ref{variational_formulation} suggests there should be a concomitant Hamiltonian formulation. In fact, the spatial dynamics variational principle fits into the standard paradigm from classical mechanics, where Lagrangian formulations correspond to Hamiltonian formulations by way of the Legendre transform. We will not discuss the Legendre transform here. Instead, we will construct an alternative Hamiltonian formulation using Lie-theoretic methods. Marsden, Ratiu, and Weinstein\cite{Marsden_SD_1984} developed the abstract theory that relates the Hamiltonian formulation developed below with the standard method based on Legendre transforms. Section \ref{sec:snakes} illustrates one of the benefits of the Lie-theoretic approach, where it is used to derive finite-dimensional Hamiltonian reductions of MHS spatial dynamics.

\begin{definition}\label{p_def}
The \textbf{Lie algebra} $\mathfrak{p}$ underlying the Hamiltonian picture comprises tuples $X=(\bm{U},S,T)\in\mathfrak{p}$, where $S$ and $T$ are smooth functions on $D$ and $\bm{U} = U^x\,\partial_x + U^y\,\partial_y$ is a vector field on $D$ that satisfies the boundary condition $x\,U^x + y\,U^y = 0$ on $\partial D$. The Lie bracket on $\mathfrak{p}$ is given by the semidirect product formula
\begin{align}
[X_1,X_2] = -([\bm{U}_1,\bm{U}_2],\mathcal{L}_{\bm{U}_1}S_2 - \mathcal{L}_{\bm{U}_2}S_1,\mathcal{L}_{\bm{U}_1}T_2 - \mathcal{L}_{\bm{U}_2}T_1).
\end{align}
\end{definition}

Phase space is $\mathfrak{p}^*$, the linear functionals $X^*:\mathfrak{p}^*\rightarrow\mathbb{R}$ over $\mathfrak{p}$. Each linear functional comprises a tuple $X^* = (\bm{U}^*,S^*,T^*)$, where $S^*$ and $T^*$ are distributional $2$-forms on $D$,
\[
S^* = \sigma\, dx\wedge dy,\quad T^* = \tau \,dx\wedge dy,
\]
and $\bm{U}^*$ is a distributional $1$-form density on $D$ that satisfies tangential boundary conditions,
\begin{align}
\bm{U}^* = (U_x\,dx + U_y\,dy)\otimes(dx\wedge dy),\quad x\,U_x + y\,U_y = 0\text{ on }\partial D.\label{one_form_density_wbc}
\end{align}
The duality pairing between $\mathfrak{p}$ and $\mathfrak{p}^*$ is
\begin{align}
\langle X^*,X\rangle =\int_D \bm{U}^*\cdot \bm{U} + \int_D S^*\,S + \int_D T^*\,T=\int_D  U_i\,U^i\,dx\,dy + \int_D \sigma\,S\,dx\,dy + \int_D \tau\,T\,dx\,dy.\nonumber
\end{align}
The Poisson bracket on $\mathfrak{p}^*$ is given by the famous Lie-Poisson formula\cite{Abraham_2008},
\begin{align}
\{F,G\}(X^*) = \left\langle X^*,\left[\frac{\delta F}{\delta X^*},\frac{\delta G}{\delta X^*}\right] \right\rangle,\quad  F,G:\mathfrak{p}^*\rightarrow\mathbb{R}.\label{LP_bracket_def}
\end{align}
Here $\delta F/\delta X^*: \mathfrak{p}^*\rightarrow\mathfrak{p}$ denotes functional differentiation, as defined by the directional derivative formula
\begin{align}
\frac{d}{d\lambda}\bigg|_0 F(X^* + \lambda\,\delta X^*) = \left\langle \frac{\delta F}{\delta X^*}(X^*),\delta X^*\right\rangle,\quad X^*,\delta X^*\in\mathfrak{p}^*.
\end{align}

We aim to recast MHS spatial dynamics as a Hamiltonian system on $\mathfrak{p}^*$, which clearly requires introducing new dependent variables. From the general theory developed by Marsden, Ratiu, and Weinstein\cite{Marsden_SD_1984}, we anticipate that $S^*$ and $T^*$ should correspond to advected quantities, while $\bm{U}^*$ should represent the momentum density arising from Legendre transformation.  Changing variables according to
\begin{align}
\bm{U}^*_\zeta = \left(\frac{\varrho_\zeta}{\omega_\zeta}\right)(\bm{u}_\zeta + \bm{N}_\zeta)^\flat\otimes \varrho_\zeta,\quad S^*_\zeta = \varrho_\zeta,\quad T^*_\zeta = p_\zeta\,\varrho_\zeta,\label{change_of_variables}
\end{align}
transforms the spatial dynamics equations \eqref{pressure_evo_ld}-\eqref{vel_evo_ld} into a from that resembles evolution equations on $\mathfrak{p}^*$ of the appropriate type. However, the $1$-form density $\bm{U}^*_\zeta$ does not necessarily satisfy the boundary conditions specified in \eqref{one_form_density_wbc}. This technical annoyance motivates restricting attention to a specific class of toroidal coordinates $(\zeta,x,y)$ where $\bm{U}^*$ does satisfy tangential boundary conditions. (The essential problem is not the form of boundary conditions on $\bm{U}^*_\zeta$ \emph{per se}, but the fact that the boundary conditions are time-dependent in general.)

The Hamiltonian formulation of MHS spatial dynamics presented here will require toroidal coordinates $(\zeta,x,y)$ that are both isothermal and boundary-compatible.
\begin{definition}
A system of toroidal coordinates $(\zeta,x,y)$ on $\widetilde{Q}$ is \textbf{boundary-compatible} if, for each point $\widetilde{q}\in\partial\widetilde{Q}$, the $\zeta$-component of the unit normal vector $\bm{n}(\widetilde{q})$ vanishes, $d\zeta(\bm{n}) = 0$. The coordinates $(\zeta,x,y)$ are \textbf{isothermal} if the hypersurface metric $h_\zeta$ is conformal to the standard flat metric $dx^2 + dy^2$, i.e. there is a nowhere vanishing function $f_\zeta$ such that $h_\zeta = f_\zeta\,(dx^2 + dy^2)$.
\end{definition}
\noindent A subsequent publication will use uniformization of surfaces to demonstrate that isothermal boundary-compatible toroidal coordinates may always be found. We will not elaborate further here. However, the following technical Lemma establishes that isothermal boundary-compatible toroidal coordinates always produce a $\bm{U}^*_\zeta$ that satisfies tangential boundary conditions.

\begin{lemma}\label{bc_lemma}
Let $(\zeta,x,y)$ denote isothermal boundary-compatible toroidal coordinates on $\widetilde{Q}$. Let $\bm{u}_\zeta$ be a vector field on $D$ tangent to $\partial D$. If $\bm{r} = x\,\partial_x + y\,\partial_y$ and $\bm{N}_\zeta$ denotes the shift vector associated with $(\zeta,x,y)$ then 
\begin{align*}
\bm{r}\cdot (\bm{u}_\zeta + \bm{N}_\zeta) = 0\text{ on }\partial D,
\end{align*}
where $\cdot$ denotes the standard dot product in $\mathbb{R}^2$.
\end{lemma}
\begin{proof}
It is clear that $\bm{r}\cdot \bm{u}_\zeta = 0$ because $\bm{u}_\zeta$ is tangent to $\partial D$. So it is enough to show that $\bm{r}\cdot \bm{N}_\zeta = 0$ on $\partial D$. 

For any system of toroidal coordinates we have $\widetilde{g}(\bm{n},\partial_\zeta) = 0$ along $\partial\widetilde{Q}$, where $\widetilde{g}$ denotes the metric tensor on $\widetilde{Q}$. But because $(\zeta,x,y)$ is boundary compatible there is a vector field $\bm{n}_\zeta = n^x_\zeta\,\partial_x + n^y\,\partial_y$ defined along $\partial D$ such that $\bm{n}(\zeta,x,y) = n^x_\zeta(x,y)\,\partial_x + n^y_\zeta(x,y)\,\partial_x$. Therefore pulling back $\widetilde{g}(\bm{n},\partial_\zeta) = 0$ along the poloidal inclusion map $I_\zeta$ implies $h_\zeta(\bm{n}_\zeta,\bm{N}_\zeta) = 0$, where $h_\zeta$ is the hypersurface metric. Since $h_\zeta= f_\zeta\,(dx^2 + dy^2)$ this in turn implies $\bm{n}_\zeta \cdot \bm{N}_\zeta = 0$. The proof will be complete if $\bm{n}_\zeta$ is parallel to $\bm{r}$, which we show next.

If $\bm{v}$ is any vector field tangent to $\partial\widetilde{Q}$ then $\widetilde{g}(\bm{v},\bm{n}) = 0$. Let $\bm{v}_\zeta(x,y) = v^x_\zeta(x,y)\,\partial_x + v^y_\zeta(x,y)\,\partial_y$ be any vector field on $D$ tangent to $\partial D$ and let $\bm{v}(\zeta,x,y) = v^x_\zeta(x,y)\,\partial_x + v^y_\zeta(x,y)\,\partial_y$ denote the corresponding vector field on $\widetilde{Q}$ that is tangent to the hypersurface foliation and $\partial \widetilde{Q}$. Pulling back $\widetilde{g}(\bm{v},\bm{n}) = 0$ along $I_\zeta$ implies $h_\zeta(\bm{v}_\zeta,\bm{n}_\zeta) = f_\zeta\, \bm{v}_\zeta\cdot \bm{n}_\zeta = 0$. Since $f_\zeta$ is nowhere-vanishing we deduce that $\bm{n}_\zeta$ must be orthogonal to $\partial D$ with respect to the standard dot product on $\mathbb{R}^2$. This implies that $\bm{n}_\zeta$ is parallel to $\bm{r}$ along $\partial D$, as desired.
\end{proof}

We may now establish the precise sense in which the MHS spatial dynamics equations \eqref{pressure_evo_ld}-\eqref{vel_evo_ld} comprise a Hamiltonian system on $\mathfrak{p}^*$.

\begin{theorem}\label{hamiltonian_formulation}
Assume $\widetilde{Q}$ is equipped with isothermal boundary-compatible toroidal coordinates $(\zeta,x,y)$. Let $\varrho_\zeta$ be an evolving nowhere-vanishing $2$-form on $D$, $\bm{u}_\zeta$ an evolving vector field on $D$, and $p_\zeta$ and evolving function on $D$. The following are equivalent.
\begin{itemize}
\item $\varrho_\zeta$, $\bm{u}_\zeta$, and $p_\zeta$ solve the spatial dynamics equations \eqref{pressure_evo_ld}-\eqref{vel_evo_ld}, with $\bm{u}_\zeta$ obeying tangential boundary conditions.
\item $\varrho_\zeta = S^*_\zeta$, $\bm{u}_\zeta =\frac{\omega_\zeta}{S^*_\zeta}(\bm{U}^*_\zeta/S^*_\zeta)^{\sharp} -\bm{N}_\zeta  $, and $p_\zeta = T^*_\zeta/S^*_\zeta$, where $(\bm{U}_\zeta^*,S^*_\zeta,T^*_\zeta)=X^*_\zeta$ solves Hamilton's equations on $\mathfrak{p}^*$ with time-dependent Hamiltonian 
\begin{align}
\mathcal{H}_\zeta(X^*)=  \int_D \left[\frac{1}{2}\left|\frac{\bm{U}^*}{S^*}\right|^2 + \frac{T^*}{S^*} - \frac{1}{2}\left(\frac{S^*}{\omega_\zeta}\right)^2N_\zeta^2\right]\omega_\zeta - \int_D\bm{U}^*\cdot \bm{N}_\zeta.\label{MHS_hamiltonian_general}
\end{align}
Here $|\cdot|$ denotes the inner product between $1$-forms defined by the hypersurface metric $h_\zeta$.
\end{itemize}
\end{theorem}
\begin{proof}
We begin by recording the functional derivatives of the Hamiltonian $\mathcal{H}_\zeta$. The calculation, which is purely algebraic, leads to the formulas
\begin{align}
\frac{\delta \mathcal{H}_\zeta}{\delta \bm{U}^*} & = \frac{\omega_\zeta}{S^*}\left(\frac{\bm{U}^*}{S^*}\right)^\sharp - \bm{N}_\zeta\\
\frac{\delta \mathcal{H}_\zeta}{\delta S^*} & = -\frac{\omega_\zeta}{S^*}\left|\frac{\bm{U}^*}{S^*}\right|^2 - \frac{\omega_\zeta}{S^*}\frac{T^*}{S^*} - \left(\frac{S^*}{\omega_\zeta}\right)\,N_\zeta^2\\
\frac{\delta \mathcal{H}_\zeta}{\delta T^*} & = \frac{\omega_\zeta}{S^*}.
\end{align}
Because the toroidal coordinates are isothermal and boundary-compatible, and $\bm{U}^*$ satisfies tangential boundary conditions by definition of $\mathfrak{p}^*$, Lemma \ref{bc_lemma} implies the vector field $\delta\mathcal{H}_\zeta/\delta \bm{U}^*$ is tangent to $\partial D$.

Next we find the explicit form of Hamilton's equations on $\mathfrak{p}^*$ with the Hamiltonian \eqref{MHS_hamiltonian_general}. Let $Q:\mathfrak{p}^*\rightarrow\mathbb{R}$ be the functional given by $Q(X^*) = \langle X^*,X\rangle$, where $X=(\bm{U},S,T)$ is a fixed, but arbitrary element of $\mathfrak{p}$. The functional derivatives of $Q$ are clearly
\begin{align*}
\frac{\delta Q}{\delta \bm{U}^*} = \bm{U},\quad \frac{\delta Q}{\delta S^*} = S,\quad \frac{\delta Q}{\delta T^*} = T.
\end{align*}
If $X^*_\zeta = (\bm{U}^*_\zeta,S^*_\zeta,T^*_\zeta)$ is a solution of Hamilton's equations on $\mathfrak{p}^*$ then $\partial_\zeta Q(X^*_\zeta) = \{Q,\mathcal{H}_\zeta\}(X^*_\zeta)$. By the chain rule and the definition of the Lie-Poisson bracket \eqref{LP_bracket_def} this implies
\begin{align}
&\int_D \partial_\zeta\bm{U}^*_\zeta\cdot \bm{U} +  \int_D \partial_\zeta S^*_\zeta\,S +\int_D \partial_\zeta T^*_\zeta\,T \nonumber\\
& =-\int_D\bm{U}^*_\zeta\cdot \left[\bm{U},\frac{\delta \mathcal{H}_\zeta}{\delta \bm{U}^*}\right] -\int_D S^*_\zeta\,\left(\mathcal{L}_{\bm{U}}\frac{\delta \mathcal{H}_\zeta}{\delta S^*} - \mathcal{L}_{\delta\mathcal{H}_\zeta/\delta\bm{U}^*}S\right)-\int_D T^*_\zeta\,\left(\mathcal{L}_{\bm{U}}\frac{\delta \mathcal{H}_\zeta}{\delta T^*} - \mathcal{L}_{\delta\mathcal{H}_\zeta/\delta\bm{U}^*}T\right).\nonumber
\end{align}
Applying Stokes theorem and the Leibniz property for the Lie derivative $\mathcal{L}_{\delta\mathcal{H}_\zeta/\delta\bm{U}^*}$ to this formula then leads to 
\begin{align}
&\int_D \partial_\zeta\bm{U}^*_\zeta\cdot \bm{U} +  \int_D \partial_\zeta S^*_\zeta\,S +\int_D \partial_\zeta T^*_\zeta\,T \nonumber\\
 & =-\int_D\left(\mathcal{L}_{\delta\mathcal{H}_\zeta/\delta\bm{U}^*}\bm{U}^*_\zeta + d\left[\frac{\delta \mathcal{H}_\zeta}{\delta S^*}\right]\otimes S^*_\zeta + d\left[\frac{\delta\mathcal{H}_\zeta}{\delta T^*}\right]\otimes T^*_\zeta\right)\cdot \bm{U} - \int_D \left(\mathcal{L}_{\delta\mathcal{H}/\delta\bm{U}^*}S^*_\zeta\right)\,S - \int_D \left(\mathcal{L}_{\delta\mathcal{H}/\delta\bm{U}^*}T^*_\zeta\right)\,T.\nonumber\\
&+ \int_{\partial D}\iota_{\delta\mathcal{H}_\zeta/\delta \bm{U}^*}\left(\bm{U}^*_\zeta\cdot \bm{U}\right) + \int_{\partial D}\iota_{\delta\mathcal{H}_\zeta/\delta \bm{U}^*}\left(S^*_\zeta\,S\right) +\int_{\partial D} \iota_{\delta\mathcal{H}_\zeta/\delta \bm{U}^*}\left(T^*_\zeta\,T\right)\label{derivation_formula}
\end{align}
Each of the boundary terms vanishes because $\delta\mathcal{H}_\zeta/\delta \bm{U}^*$ is tangent to $\partial D$. Since $X\in\mathfrak{p}$ is arbitrary we therefore find
\begin{align}
\partial_\zeta\bm{U}^*_\zeta &= -\mathcal{L}_{\delta\mathcal{H}_\zeta/\delta\bm{U}^*}\bm{U}^*_\zeta - d\left[\frac{\delta \mathcal{H}_\zeta}{\delta S^*}\right]\otimes S^*_\zeta - d\left[\frac{\delta\mathcal{H}_\zeta}{\delta T^*}\right]\otimes T^*_\zeta\label{ham_momentum}\\
\partial_\zeta S^*_\zeta &= -\mathcal{L}_{\delta\mathcal{H}/\delta\bm{U}^*}S^*_\zeta\label{ham_S}\\
\partial_\zeta T^*_\zeta & = -\mathcal{L}_{\delta\mathcal{H}/\delta\bm{U}^*}T^*_\zeta.\label{ham_T}
\end{align}
The equations \eqref{ham_momentum}-\eqref{ham_T} express Hamilton's equations on $\mathfrak{p}^*$ for the Hamiltonian $\mathcal{H}_\zeta$ given in \eqref{MHS_hamiltonian_general}. Equivalently, substituting the above formulas for functional derivatives of $\mathcal{H}_\zeta$ leads to
\begin{gather}
(\partial_\zeta + \mathcal{L}_{\bm{v}_\zeta})\bm{U}^*_\zeta = d\left[\frac{T^*_\zeta}{S^*_\zeta}\right]\otimes \omega_\zeta+d\left[\left(\frac{\omega_\zeta}{S^*_\zeta}\right)\left|\frac{\bm{U}^*_\zeta}{S^*_\zeta}\right|^2 + \left(\frac{S^*_\zeta}{\omega_\zeta}\right)\,N_\zeta^2\right]\otimes S^*_\zeta,\quad  \bm{v}_\zeta  = \frac{\omega_\zeta}{S^*_\zeta}\left(\frac{\bm{U}^*_\zeta}{S^*_\zeta}\right)^\sharp - \bm{N}_\zeta\label{ham_momentum_explicit}\\
(\partial_\zeta + \mathcal{L}_{\bm{v}_\zeta})S^*_\zeta = 0\label{ham_S_explicit},\\
(\partial_\zeta + \mathcal{L}_{\bm{v}_\zeta})T^*_\zeta = 0\label{ham_T_explicit}.
\end{gather}

Next we observe how Hamilton's equations \eqref{ham_momentum_explicit}-\eqref{ham_T_explicit} transform when subjected to the following invertible change of dependent variables, $(\bm{U}^*_\zeta,S^*_\zeta,T^*_\zeta)\mapsto (\bm{u}_\zeta,\varrho_\zeta,p_\zeta)$:
\begin{align*}
\bm{U}^*_\zeta = \left(\frac{S^*_\zeta}{\omega_\zeta}\right)(\bm{u}_\zeta + \bm{N}_\zeta)^\flat\otimes \varrho_\zeta,\quad S^*_\zeta = \varrho_\zeta,\quad T^*_\zeta = p_\zeta\,\varrho_\zeta.
\end{align*}
The Hamilton equation \eqref{ham_S_explicit} immediately implies $(\partial_\zeta + \mathcal{L}_{\bm{v}_\zeta})\varrho_\zeta = 0$. The nowhere vanishing property for $S^*_\zeta = \varrho_\zeta$, the Leibniz property for $\partial_\zeta + \mathcal{L}_{\bm{v}_\zeta}$, and the Hamilton equation \eqref{ham_T_explicit} then gives $(\partial_\zeta + \mathcal{L}_{\bm{v}_\zeta})p_\zeta = 0$. The Leibniz property for $\partial_\zeta + \mathcal{L}_{\bm{v}_\zeta}$ and the advection equation for $\varrho_\zeta$ also imply
\begin{align*}
(\partial_\zeta + \mathcal{L}_{\bm{v}_\zeta})\bm{U}^*_\zeta = \left[(\partial_\zeta + \mathcal{L}_{\bm{v}_\zeta})\pi_\zeta\right]\otimes \varrho_\zeta,\quad \pi_\zeta = \left(\frac{S^*_\zeta}{\omega_\zeta}\right)(\bm{u}_\zeta + \bm{N}_\zeta)^\flat.
\end{align*}
The Hamilton equation \eqref{ham_momentum_explicit} may therefore be rewritten as
\begin{align*}
\left[(\partial_\zeta + \mathcal{L}_{\bm{v}_\zeta})\pi_\zeta\right]\otimes \varrho_\zeta = \left(\frac{\omega_\zeta}{\varrho_\zeta}\right)dp_\zeta\otimes \varrho_\zeta+d\left[\left(\frac{\omega_\zeta}{\varrho_\zeta}\right)\left|\pi_\zeta\right|^2 + \left(\frac{\varrho_\zeta}{\omega_\zeta}\right)\,N_\zeta^2\right]\otimes \varrho_\zeta,
\end{align*}
or, upon division by $\varrho_\zeta$,
\begin{align*}
(\partial_\zeta + \mathcal{L}_{\bm{v}_\zeta})\pi_\zeta = \left(\frac{\omega_\zeta}{\varrho_\zeta}\right)dp_\zeta+d\left[\left(\frac{\omega_\zeta}{\varrho_\zeta}\right)\left|\pi_\zeta\right|^2 + \left(\frac{\varrho_\zeta}{\omega_\zeta}\right)\,N_\zeta^2\right].
\end{align*}
Finally, notice that
\begin{align*}
\bm{v}_\zeta = \frac{\omega_\zeta}{S^*_\zeta}\left(\frac{\bm{U}^*_\zeta}{S^*_\zeta}\right)^\sharp - \bm{N}_\zeta =\left((\bm{u}_\zeta + \bm{N}_\zeta)^\flat\right)^\sharp - \bm{N}_\zeta = \bm{u}_\zeta, 
\end{align*}
and that $\bm{u}_\zeta = \delta\mathcal{H}_\zeta/\delta\bm{U}^*$ satisfies tangential boundary conditions. 

These remarks show that Eqs.\,\eqref{ham_momentum_explicit}-\eqref{ham_T_explicit} are related to the MHS spatial dynamics equations \eqref{pressure_evo_ld}-\eqref{vel_evo_ld} by the invertible change of dependent variables specified in the second bullet of the theorem statement. The theorem now follows immediately.

\end{proof}

\section{Manifestation of mixed elliptic-hyperbolic type in MHS spatial dynamics\label{sec:fs_formulation}}
The MHS equations belie simplicity due in part to their well-known mixed elliptic-hyperbolic type. As written in their usual form, $(\nabla\times\bm{B})\times\bm{B} = \nabla p$, $\nabla\cdot\bm{B} = 0$, the elliptic and hyperbolic aspects of the equations intertwine in a befuddling manner. This Section shows that the spatial dynamics formulation of MHS offers a suggestive dynamical-systems interpretation of type mixing that may prove useful in future rigorous investigations of existence theory for non-symmetric MHS solutions, perhaps along similar lines to those explored by by Mielke\cite{Mielke_1992}. 

Consider the spatial dynamics formulation of MHS in a ``large aspect ratio" volume of revolution $Q$. Qualitatively, large aspect ratio means $Q$ looks more like a bicycle tire than a cored apple; this notion will be made precise as the discussion continues. Without loss of generality, assume the toroidal coordinates $(\zeta,x,y):Q\rightarrow S^1\times D$ enjoy the following properties: (1) the toroidal angle $\zeta:Q\rightarrow S^1$ is equal to the azimuthal angle from a standard cylindrical coordinate system adapted to $Q$; (2) the hypersurface metric $h_\zeta= h$ is independent of $\zeta$. The first assumption implies that the vector fields $\partial_x,\partial_y$ are each orthogonal to $\partial_\zeta$, which in turn implies that the shift form (and its corresponding shift vector) $\mathcal{N}_\zeta = 0$ vanishes. It also implies that the lapse function $N_\zeta = N$ is independent of $\zeta$. The hypersurface area element $\omega_\zeta$ therefore becomes  $\omega_\zeta = N\,\sigma$,
where $\sigma = \sqrt{\text{det}\,h}\,dx\wedge dy$ is the Riemannian area form on $D$ defined by the hypersurface metric $h$. In light of all this, the MHS spatial dynamics equations on $\widetilde{Q}$ reduce to
\begin{gather*}
\partial_\zeta\pi_\zeta + \iota_{\bm{u}_\zeta}d\pi_\zeta =\frac{N^2}{F_\zeta} dp_\zeta + dF_\zeta,\quad (\partial_\zeta + \mathcal{L}_{\bm{u}_\zeta})\varrho_\zeta = 0,\quad (\partial_\zeta + \mathcal{L}_{\bm{u}_\zeta})p_\zeta = 0,\quad \iota_{\bm{u}_\zeta}h = \frac{N\sigma}{\varrho_\zeta}\pi_\zeta,\quad F_\zeta = \frac{N\,\varrho_\zeta}{\sigma},
\end{gather*}
together with the boundary condition $\pi_\zeta(\bm{n})=0$ on $\partial D$. If $\star$ denotes the Hodge star associated with $h$ and $t:\partial D\rightarrow D$ denotes the boundary inclusion then the boundary condition on $\pi_\zeta$ is equivalent to $t^*(\star\pi_\zeta) = 0$. Physically the boundary $\partial Q$ should be an isobar. By force balance, if the boundary is an isobar then the normal component of current density $\bm{n}\cdot\nabla\times\bm{B}$ must vanish along $\partial Q$. Therefore we impose the following pair of additional boundary conditions
\begin{align}
t^*p_\zeta = \text{const.},\quad t^*(\partial_\zeta\pi_\zeta - dF_\zeta) = 0,\label{isobar_bcs}
\end{align}
which we will refer to as \textbf{isobaric boundary conditions}. Note that isobaric boundary conditions are implied by the evolution equations for $\pi_\zeta$ and $p_\zeta$ if one requires the latter are satisfied on the boundary $\partial D$. These additional boundary conditions are usually taken for granted because $\nabla\times\bm{B}\cdot\nabla p = (\nabla\times\bm{B})\cdot([\nabla\times\bm{B}]\times\bm{B}) =0$ in equilibrium. But they will play a surprisingly important role in the formal analysis that follows.

Introduce a small parameter $\epsilon$ that stands for the inverse aspect ratio of $Q$; $\epsilon$ represents the ratio of a bike tire's small radius to its large radius. Make the large-aspect-ratio assumption precise by defining dimensionless variants of the quantities that specify the geometry of the hypersurface foliation according to
\begin{align*}
h = \epsilon^2\,R_0^2\,\mathsf{h},\quad N = R_0\,\mathsf{N}_\epsilon = R_0(1+\epsilon\,\mathsf{N}_1),\quad \sigma=\epsilon^2\,R_0^2\,\Sigma.
\end{align*}
The dimensionless hypersurface metric, lapse function, and Riemannian hypersurface area element are $\mathsf{h},\mathsf{N}_\epsilon$, and $\Sigma$, respectively. The real constant $R_0 > 0$ denotes the major radius of $Q$; it should be thought of as prescribing the larger of the two bike tire radii. The dimensional hypersurface metric $h$ scales like $\epsilon^2$ because it measures squared lengths within toroidal slices $D_\zeta$, and the diameter of each of these slices should scale like $\epsilon\,R_0$ in the large aspect ratio regime. The scaling of $\sigma$ ensures that $\Sigma$ is the Riemannian volume element associated with $\mathsf{h}$. The scaling of $N$ ensures the geometry of $Q$ asymptotically approaches that of a straight cylinder as $\epsilon\rightarrow 0$. Observe that the dimensionless hypersurface metric $\mathsf{h}$ defines a dimensionless Hodge-$\smallstar$ operator $\smallstar$ that relates to the dimensional Hodge-$\star$ according to
\begin{align*}
\star = \begin{cases} \epsilon^2\,R_0^2\,\smallstar &\text{ on $0$-forms}\\
\smallstar& \text{ on $1$-forms}\\
 \frac{1}{\epsilon^2\,R_0^2}\smallstar& \text{ on $2$-forms}.
\end{cases}
\end{align*}
The dimensional $h$-Laplacian operator on functions, $\nabla^2 =\star d \star d$, therefore relates to the dimensionless $\mathsf{h}$-Laplacian $\Delta = \smallstar d \smallstar d$ according to $\nabla^2 = (\epsilon\,R_0)^{-2}\Delta$.

We would like to formally study how solutions of the spatial dynamics equations behave when $\epsilon << 1$. For this it is useful to exchange the dimensional dependent variables $(\pi_\zeta,\varrho_\zeta,p_\zeta)$ with a new set of scalar dimensionless dependent variables $(A_\zeta,\chi_\zeta,f_\zeta,\mathsf{p}_\zeta)$ defined by
\begin{align*}
\pi_\zeta = \epsilon^2\,B_0\,R_0\,(\smallstar d A_\zeta + d\chi_\zeta),\quad \varrho_\zeta = \epsilon^2\,B_0\,R_0^2\frac{1+\epsilon\,f_\zeta}{1+\epsilon\,N_1}\Sigma,\quad p_\zeta = \epsilon^2\,B_0^2\,\mathsf{p}_\zeta,
\end{align*}
together with the boundary conditions $t^*A_\zeta = 0$ and $t^*(\smallstar d\chi_\zeta) = 0$. The real constant $B_0 > 0$ represents a typical magnetic field strength. Invertibility of the transformation follows from the Hodge-Morrey-Friedrichs decomposition\cite{Schwarz_1995} for $1$-forms described in detail by G. Schwarz. The scalings of each dimensional dependent variable with $\epsilon$, $B_0$, and $R_0$ correspond to the original (low-$\beta$) reduced MHD\cite{Strauss_1975,Morrison_1984} orderings introduced by H. Strauss. The corresponding fully-dimensionless form of the spatial dynamics equations is
\begin{gather}
\partial_\zeta \Delta A_\zeta + \left\langle \frac{(1+\epsilon\mathsf{N}_1)^2}{1+\epsilon\,f_\zeta}\Delta A_\zeta,\chi_\zeta \right\rangle - \left\{ \frac{(1+\epsilon\mathsf{N}_1)^2}{1+\epsilon\,f_\zeta}\Delta A_\zeta,A_\zeta\right\} + \frac{(1+\epsilon\mathsf{N}_1)^2}{1+\epsilon\,f_\zeta}\Delta A_\zeta\,\Delta \chi_\zeta = \left\{\frac{(1+\epsilon\mathsf{N}_1)^2}{1+\epsilon\,f_\zeta},\mathsf{p}_\zeta\right\}\label{poloidal_flux_evo}\\
\partial_\zeta\Delta \chi_\zeta - \left\langle \frac{(1+\epsilon\mathsf{N}_1)^2}{1+\epsilon\,f_\zeta}\Delta A_\zeta,A_\zeta\right\rangle - \left\{\frac{(1+\epsilon\mathsf{N}_1)^2}{1+\epsilon\,f_\zeta}\Delta A_\zeta,\chi_\zeta\right\} - \frac{(1+\epsilon\mathsf{N}_1)^2}{1+\epsilon\,f_\zeta}(\Delta A_\zeta)^2\nonumber\\
= \left\langle \frac{(1+\epsilon\mathsf{N}_1)^2}{1+\epsilon\,f_\zeta},\mathsf{p}_\zeta \right\rangle + \frac{(1+\epsilon\mathsf{N}_1)^2}{1+\epsilon\,f_\zeta}\Delta \mathsf{p}_\zeta + \frac{1}{\epsilon}\Delta f_\zeta\label{magnetic_potential_evo}\\
\epsilon\,\partial_\zeta f_\zeta + \frac{1}{2}\left\langle (1+\epsilon\,\mathsf{N}_1)^2,\chi_\zeta\right\rangle - \frac{1}{2}\left\{ (1+\epsilon\,\mathsf{N}_1)^2,A_\zeta\right\} + (1+\epsilon\,\mathsf{N}_1)^2\,\Delta\chi_\zeta = 0\label{F_evo}\\
\partial_\zeta\mathsf{p}_\zeta + \frac{(1+\epsilon\mathsf{N}_1)^2}{1+\epsilon\,f_\zeta}\,\{A_\zeta,\mathsf{p}_\zeta\} + \frac{(1+\epsilon\mathsf{N}_1)^2}{1+\epsilon\,f_\zeta}\left\langle \chi_\zeta,\mathsf{p}_\zeta \right\rangle=0.\label{pressure_evo}
\end{gather}
Here the Poisson bracket $\{\cdot,\cdot\}$ and metric bracket $\langle\cdot,\cdot\rangle$ are defined according to
\begin{align*}
\{f,g\} = \smallstar (df\wedge dg),\quad \langle f,g\rangle = \smallstar(df\wedge \smallstar dg).
\end{align*}
The partial differential equations \eqref{poloidal_flux_evo}-\eqref{pressure_evo} on $D$ should be supplemented with the bevy of boundary conditions
\begin{align}
t^*A_\zeta = 0,\quad t^*\smallstar d\chi_\zeta = 0,\quad \int_D \chi_\zeta\,\Sigma = 0,\quad t^*\mathsf{p}_\zeta = \text{const.},\quad t^*(\epsilon\,\smallstar d\partial_\zeta A_\zeta + \epsilon\,d\partial_\zeta\chi_\zeta -df_\zeta ) = 0.
\end{align}
The first three stem from the potential representation of $\pi_\zeta$ together with the tangential boundary condition $\pi_\zeta(\bm{n}) = 0$, while the last two represent isobaric boundary conditions expressed in terms of the potentials.

The dynamical equations \eqref{poloidal_flux_evo}-\eqref{pressure_evo} exhibit two disparate timescales. The short-timescale behavior can be isolated by exchanging slow time $\zeta$ with fast time $s = \zeta/\epsilon$ and taking the limit $\epsilon\rightarrow 0$. The result is
\begin{gather}
\partial_s \Delta A_s = 0,\quad \partial_s \Delta \chi_s = \Delta f_s,\quad \partial_s f_s = -\Delta \chi_s,\quad \partial_s \mathsf{p}_s = 0,\label{complete_fast_equations}
\end{gather}
subject to the boundary conditions
\begin{align*}
t^*A_s = 0,\quad t^*\smallstar d\chi_s = 0,\quad \int_D \chi_s\,\Sigma = 0,\quad t^*\mathsf{p}_s = \text{const.},\quad t^*(\smallstar d\partial_s A_s + d\partial_s\chi_s -df_s ) = 0.
\end{align*}
These equations admit $A_s$, $\mathsf{p}_s$, and $\gamma_s = \int_Df_s\Sigma$ as constants of motion. The remaining dynamical variables $\widetilde{f}_s = f_s - \gamma_s$ and $\chi_s$ evolve according to 
\begin{align}
\partial_s\chi_s = \widetilde{f}_s,\quad \partial_s \widetilde{f}_s = -\Delta\chi_s,\quad \int_D\widetilde{f}_s\Sigma = \int_D \chi_s\Sigma = 0,\quad t^*(\smallstar d\chi_s) = 0.\label{fast_dynamics_equations}
\end{align}
(Note we have used $\partial_s A_s = 0$ and $\int_D\chi_s\Sigma = 0$ to simplify the isobaric boundary conditions to $t^*(\partial_s\chi_s-\widetilde{f}_s) = 0$.) The equations \eqref{fast_dynamics_equations} have a unique fixed point $\widetilde{f}_s = \chi_s = 0$. Therefore the complete set of fast evolution equations \eqref{complete_fast_equations} contains an infinite-dimensional vector space of fixed points $E_0$ parameterized by the constants of motion $A_0,\mathsf{p}_0$ and $\gamma_0$. The stability type of these fixed points normal to $E_0$ may be determined by solving \eqref{fast_dynamics_equations} for general initial data. Differentiating the $\chi_s$ evolution equation in $s$ reveals that $\chi_s$ must obey a wave equation with imaginary propagation speed, i.e. the Poisson equation $\partial_s^2\chi_s = -\Delta \chi_s$. The initial value problem for this equation can be solved explicitly in terms of Neumann eigenfunctions for $\Delta$ on $D$. The resulting solutions contain exponentially growing and exponentially decaying modes with growth rate proportional to eigenvalues of $\Delta$. Thus, the space of fixed points $E_0$ is formally normally hyperbolic, and we may say that the fast part of spatial dynamics in the large-aspect-ratio regime corresponds to the elliptic part of the MHS equations. 

Notice that the initial value problem for the fast equations is ill-posed! Such ill-posedness is neither surprising nor fatal; spatial dynamics formulations of elliptic equations characteristically lead to ill-posed evolution equations\cite{Beck_2020}. Overall, the large-aspect-ratio MHS spatial dynamics equations comprise an infinite-dimensional analogue of the singularly-perturbed dynamical systems covered by Fenichel's\cite{Fenichel_1979} geometric singular perturbation theory, which says in particular that when $0<\epsilon << 1$ the spatial dynamics equations should admit a slow normally-hyperbolic invariant manifold given to leading order by $E_0$. The natural way to address the ill-posedness of the MHS spatial dynamics equations is restricting initial conditions to this slow manifold. However, as the formal calculations below will reveal, the spatial dynamics slow manifold must be infinite-dimensional. The analytic theory required to establish existence of such a slow manifold is not currently available. On an encouraging note,  Mielke\cite{Mielke_1992} studied infinite-dimensional center manifolds (slow manifolds, properly interpreted, are center manifolds) arising in a class of partial differential equations with mixed elliptic-hyperbolic type, but his theory does not appear to cover MHS. Infinite-dimensional center manifolds, although with finitely-many stable and unstable directions, also form the crux of analysis presented by Groves and Schneider\cite{Groves_Schneider_2001,Groves_Schneider_2005,Groves_Schneider_2008} in a series of papers studying modulated pulse solutions of nonlinear wave equations.

One immediate consequence of the preceding formal analysis of the fast evolution equations \eqref{complete_fast_equations} is existence of a formal slow manifold for the full spatial dynamics equations \eqref{poloidal_flux_evo}-\eqref{pressure_evo} in the sense described in the review article\cite{Burby_Klotz_2020} by Burby and Klotz.  Reduction of the spatial dynamics equations to this slow manifold proceeds as follows. As a minor preparatory step, decompose $f_\zeta$ as $f_\zeta =(1+\epsilon\,\mathsf{N}_1) \gamma_\zeta + \widetilde{f}_\zeta$, where $\int_D \widetilde{f}_\zeta(1+\epsilon\,\mathsf{N}_1)^{-1}\Sigma = 0$; it is simple to show that $\gamma_s\in\mathbb{R}$ is a constant of motion for all values of $\epsilon$. Suppose there is an invariant set $E_\epsilon$ for the evolution equations \eqref{poloidal_flux_evo}-\eqref{pressure_evo} given as a graph over the space of fixed points $E_0$ for the limiting fast evolution. This set $E_\epsilon$ will be the slow manifold. The graph property means there are functionals $\chi^*_\epsilon(A,\mathsf{p},\gamma)$ and $\widetilde{f}^*_\epsilon(A,\mathsf{p},\gamma)$ with $t^*\smallstar d\chi^*_\epsilon = 0$ and $\int \widetilde{f}^*_\epsilon\,(1+\epsilon\,\mathsf{N}_1)^{-1}\Sigma = 0$ such that
\begin{align*}
E_\epsilon = \{(A,\chi,\gamma,\widetilde{f},\mathsf{p})\mid \chi = \chi^*_\epsilon(A,\mathsf{p},\gamma),\quad \widetilde{f} = \widetilde{f}^*_\epsilon(A,\mathsf{p},\gamma)\}.
\end{align*}
Invariance means that if $(A_\zeta,\chi_\zeta,\gamma_\zeta,\widetilde{f}_\zeta,\mathsf{p}_\zeta)$ is a solution of the spatial dynamics equations \eqref{poloidal_flux_evo}-\eqref{pressure_evo} that starts in $E_\epsilon$ then that solution will remain in $E_\epsilon$ for all subsequent times. Solutions that start in $E_\epsilon$ therefore enjoy the special property that the \textbf{fast variables} $\chi_\zeta$ and $\widetilde{f}_\zeta$ are completely determined by the \textbf{slow variables} $A_\zeta,\gamma_\zeta,\mathsf{p}_\zeta$ according to the \textbf{slaving relations}
\begin{align*}
\chi_\zeta = \chi^*_\epsilon(A_\zeta,\mathsf{p}_\zeta,\gamma_\zeta),\quad \widetilde{f}_\zeta
 = \widetilde{f}^*_\epsilon(A_\zeta,\mathsf{p}_\zeta,\gamma_\zeta).
\end{align*}
Substituting the slaving relations into the fast variable evolution equations \eqref{magnetic_potential_evo} and \eqref{F_evo} leads to the functional partial differential equation called the \textbf{invariance equation} that the functionals $\chi^*_\epsilon$ and $\widetilde{f}_\epsilon^*$ must obey:
\begin{gather*}
\Delta \bigg(D\chi^*_\epsilon[\dot{A},\dot{\mathsf{p}},\dot{\gamma}]\bigg) - \left\langle \frac{(1+\epsilon\mathsf{N}_1)^2}{\mathsf{F}^*_\epsilon}\Delta A,A\right\rangle - \left\{\frac{(1+\epsilon\mathsf{N}_1)^2}{\mathsf{F}^*_\epsilon}\Delta A,\chi^*_\epsilon\right\} - \frac{(1+\epsilon\mathsf{N}_1)^2}{\mathsf{F}^*_\epsilon}(\Delta A)^2\nonumber\\
= \left\langle \frac{(1+\epsilon\mathsf{N}_1)^2}{\mathsf{F}^*_\epsilon},\mathsf{p} \right\rangle + \frac{(1+\epsilon\mathsf{N}_1)^2}{\mathsf{F}^*_\epsilon}\Delta \mathsf{p} 
+\gamma\,\Delta \mathsf{N}_1 + \frac{1}{\epsilon}\,\Delta \widetilde{f}^*_\epsilon,\quad \mathsf{F}^*_\epsilon = 1+\epsilon\,\gamma + \epsilon\,\widetilde{f}^*_\epsilon + \epsilon^2\,\gamma\,\mathsf{N}_1,\label{magnetic_potential_invariance}\\
\epsilon\,D\widetilde{f}^*_\epsilon[\dot{A},\dot{\mathsf{p}},\dot{\gamma}] + \frac{1}{2}\left\langle (1+\epsilon\,\mathsf{N}_1)^2,\chi^*_\epsilon\right\rangle - \frac{1}{2}\left\{ (1+\epsilon\,\mathsf{N}_1)^2,A\right\} + (1+\epsilon\,\mathsf{N}_1)^2\,\Delta\chi^*_\epsilon = 0,\label{F_invariance}
\end{gather*}
together with the boundary conditions
\begin{align*}
t^*(\smallstar d\chi^*_\epsilon) = 0,\quad \int_D \frac{\widetilde{f}^*_\epsilon}{1+\epsilon\,\mathsf{N}_1}\Sigma = 0,\quad t^*(\epsilon\,\smallstar d\dot{A} + \epsilon\,dD\chi^*_\epsilon[\dot{A},\dot{\mathsf{p}},\dot{\gamma}]-\epsilon\,\gamma\,d\mathsf{N}_1 -d\widetilde{f}^*_\epsilon ) = 0.
\end{align*}
Here $\dot{A},\dot{\mathsf{p}},\dot{\gamma}$ denote time derivatives of $A,\mathsf{p},\gamma$ expressed in terms of those fields as well as the functionals $\chi^*_\epsilon,\widetilde{f}^*_\epsilon$. Now formal solutions for the functionals $\chi^*_\epsilon$ and $\widetilde{f}^*_\epsilon$ can be obtained by substituting the formal power series expansions $\chi^*_\epsilon = \chi^*_0 + \epsilon\,\chi^*_1 + \dots$, $\widetilde{f}^*_\epsilon = \widetilde{f}^*_0 + \epsilon\,\widetilde{f}^*_1 + \dots$ into the invariance equation and solving order-by-order in $\epsilon$. With the functionals determined, the slow manifold reduction of the MHS spatial dynamics equations is then given, to all orders in $\epsilon$, by 
\begin{gather}
\partial_\zeta \Delta A_\zeta + \left\langle \frac{(1+\epsilon\mathsf{N}_1)^2}{\mathsf{F}^*_\zeta}\Delta A_\zeta,\chi^*_\epsilon \right\rangle - \left\{ \frac{(1+\epsilon\mathsf{N}_1)^2}{\mathsf{F}^*_\epsilon}\Delta A_\zeta,A_\zeta\right\} + \frac{(1+\epsilon\mathsf{N}_1)^2}{\mathsf{F}^*_\epsilon}\Delta A_\zeta\,\Delta \chi^*_\epsilon = \left\{\frac{(1+\epsilon\mathsf{N}_1)^2}{\mathsf{F}^*_\epsilon},\mathsf{p}_\zeta\right\}\label{poloidal_flux_slow}\\
\partial_\zeta\mathsf{p}_\zeta + \frac{(1+\epsilon\mathsf{N}_1)^2}{\mathsf{F}^*_\epsilon}\,\{A_\zeta,\mathsf{p}_\zeta\} + \frac{(1+\epsilon\mathsf{N}_1)^2}{\mathsf{F}^*_\epsilon}\left\langle \chi^*_\epsilon,\mathsf{p}_\zeta \right\rangle=0.\label{pressure_slow}\\
\partial_\zeta\gamma_\zeta = 0.
\end{gather}
Given a ``solution" $(A_\zeta,\mathsf{p}_\zeta,\gamma_\zeta)$ of these slow manifold reduced equations, a solution of the full MHS spatial dynamics system can be formally recovered by reconstructing the fast variables according to $\widetilde{f}_\zeta = \widetilde{f}^*_\epsilon(A_\zeta,\mathsf{p}_\zeta,\gamma_\zeta)$ and $\chi_\zeta = \chi^*_\epsilon(A_\zeta,\mathsf{p}_\zeta,\gamma_\zeta)$. But practically speaking the functionals $\chi^*_\epsilon,\widetilde{f}^*_\epsilon$ need to be truncated at some order in $\epsilon$. After truncation, solutions of the slow manifold reduced equations only give rise to approximate solutions of the full spatial dynamics system. This procedure gives rise to a new notion of approximate solution of the MHS equations in large-aspect-ratio domains.

It is straightforward to find the first two coefficients in the formal series expansions for $\chi^*_\epsilon$ and $\widetilde{f}^*_\epsilon$. Computing higher-order coefficients by similar means is also straightforward, in principle, but involves increasingly complex algebraic manipulations. To leading-order, the invariance equation is
\begin{align*}
\Delta \widetilde{f}^*_0 = 0,\quad \Delta\chi^*_0 = 0,\quad \int_D \widetilde{f}^*_0\Sigma =\int_D \chi^*_0\,\Sigma= 0,\quad t^*(\smallstar d\chi^*_0)=0,\quad t^*(d\widetilde{f}^*_0) = 0,
\end{align*}
which implies $\widetilde{f}^*_0 = 0$ and $\chi^*_0 = 0$. At next order in $\epsilon$, the invariance equation therefore becomes
\begin{gather}
\Delta \widetilde{f}^*_1 = -\Delta \mathsf{p} - \gamma\,\Delta\mathsf{N}_1 - \langle\Delta A,A\rangle- (\Delta A)^2,\quad \int_D\widetilde{f}^*_1\Sigma = 0,\quad t^*d(\widetilde{f}_1^* + \gamma\mathsf{N}_1 - C) = 0\\
\Delta \chi^*_1 = \{\mathsf{N}_1,A\},\quad \int_D\chi^*_1\Sigma = 0,\quad t^*(\smallstar d\chi^*_1) = 0,\label{chi1_eqn}
\end{gather}
where the function $C$ is the unique solution of the boundary value problem
\begin{gather*}
\Delta C = \smallstar d\smallstar(A\,d\Delta A),\quad \int_DC\Sigma =0,\quad t^*(\smallstar dC)=0.
\end{gather*}
Again there is a unique solution for $\widetilde{f}^*_1$ and $\chi^*_1$, but now it can only be expressed as the solution of a pair of Poisson equations with standard boundary conditions. This pattern persists as the calculation is continued to higher-order in $\epsilon$, vividly reemphasizing the relationship between the slow manifold and the elliptic part of the MHS equations. The field $C$ arises from the following consideration. The leading-order evolution equation for $A$ reads $\Delta \dot{A}_0 + \{A,\Delta A\} = 0$, which is equivalent to $d(\smallstar d\dot{A}_0 + A\,d\Delta A) =0$, which implies in turn $\smallstar d\dot{A}_0 + A\,d\Delta A = dC$ for some function $C$ uniquely defined up to a constant. We are free to fix the constant by requiring $\int_D C\Sigma= 0$. Moreover, because $t^*A = t^*\dot{A}_0 = 0$, the function $C$ must satisfy homogeneous Neumann boundary conditions $t^*(\smallstar dC) = 0$. The above boundary value problem for $C$ now follows from applying $\smallstar d\smallstar$ to the equation $\smallstar d\dot{A}_0 + A\,d\Delta A = dC$.

With the first two coefficients in the series expansion for $\widetilde{f}^*_\epsilon$ and $\chi^*_\epsilon$ in hand, the slow manifold reduced spatial dynamics equations may now be expressed with $O(\epsilon)$ accuracy. We find
\begin{gather*}
\partial_\zeta\Delta A_\zeta + (1-\epsilon\,\gamma_\zeta + 2\epsilon\,\mathsf{N}_1)\{A_\zeta,\Delta A_\zeta\} = 2\epsilon\{\mathsf{N}_1,\mathsf{p}_\zeta\} + \epsilon\,\Delta A_\zeta\,\{\mathsf{N}_1,A_\zeta\}\\
\partial_\zeta\mathsf{p}_\zeta + (1-\epsilon\,\gamma_\zeta + 2\epsilon\,\mathsf{N}_1)\,\{A_\zeta,\mathsf{p}_\zeta\} + \epsilon\,\left\langle \chi^*_1,\mathsf{p}_\zeta \right\rangle=0\\
\partial_\zeta\gamma_\zeta = 0,
\end{gather*}
where $\chi^*_1$ denotes the solution of the Neumann Poisson equation \eqref{chi1_eqn}.
This truncated form of the $O(\epsilon)$ slow manifold reduction of MHS spatial dynamics is likely not a Hamiltonian system because the truncation scheme used to find it did not make use of the Hamiltonian structure underlying MHS spatial dynamics. It would be interesting to develop Hamiltonian truncations for higher-order slow manifold reductions of the spatial dynamics equations along lines pursued previously in dynamical MHD\cite{Burby_TF_2017} and the weakly-relativistic kinetic plasma model\cite{Miloshevich_2021}. When $\epsilon = 0$, the slow manifold reduced equations degenerate to
\begin{gather*}
\partial_\zeta\Delta A_\zeta + \{A_\zeta,\Delta A_\zeta\} = 0\\
\partial_\zeta\mathsf{p}_\zeta + \{A_\zeta,\mathsf{p}_\zeta\} =0\\
\partial_\zeta\gamma_\zeta = 0.
\end{gather*}
Hyperbolicity of these equations reveals that the hyperbolic part of the MHS equations corresponds to motion along the slow manifold in the spatial dynamics framework. Notice that the limiting equation for $A_\zeta$ coincides with Strauss'\cite{Strauss_1975} Eq.\,(9) in the special case of equilibrium and vanishing flow, and also the vorticity form of the 2D planar Euler equations.

\section{Smoothed-particle MHS\label{sec:snakes}}
MHS spatial dynamics formally resembles time-dependent planar hydrodynamics. Methods used to simplify analysis of planar fluids therefore suggest analogues applicable to MHS. For example, smoothed-particle hydrodynamics (SPH) {\color{red}} replaces a fluid continuum with finitely-many mass lumps, or particles, that move like Lagrangian tracers. This suggests replacing a solution of the MHS spatial dynamics equations with finitely-many flux lumps that move along magnetic field lines. This Section uses the Hamiltonian structure of MHS spatial dynamics from Section \ref{sec:lie_poisson} to construct an SPH-like approximation of the spatial dynamics equations that we call smoothed-particle magnetohydrostatics (SPMHS). In fact, our derivation of SPMHS shows it describes exact particle-like solutions of a physically-plausible regularization of the MHS equations. 

Particles in SPMHS comprise point-like representatives of narrow bundles of magnetic field lines, much in the way marker particles from the particle-in-cell plasma simulation method\cite{Dawson_1983,Squire_2012,Evstatiev_2013,Kraus_2017,Burby_fdk_2017,Xiao_2018,Li_2019,Pinto_2022} represent mesoscopic ensembles of actual plasma particles. The following intuitive procedure constructs the $N$-particle representation of a given smooth distribution of toroidal flux. Decompose the unit disc $D$ into $N$ disjoint cells $C_i\subset D$, $i=1,\dots, N$, each with diameter $\alpha$. Assume that $N$ is large enough so the magnetic field strength $B$ over each cell $C$ is roughly constant. Since the union of all cells gives $D$, the cell diameter scales like $\alpha \sim 1/\sqrt{N}$. So as the number of cells $N$ increases the toroidal flux $B\alpha^2 \sim B/N$ and the magnetic energy $B^2\alpha^2 \sim B^2/N$ of each cell becomes vanishingly small while keeping the total flux and energy constant. The centers $\bm{q}$ of the cells $C$ for large $N$ correspond to SPMHS particles.

The world line of an SPMHS particle corresponds to a flux tube. The idea of representing solutions of MHD equations using flux tubes dates back at least to the work of Spruit\cite{Spruit_1981}, who proposed equations of motion for a thin isolated circular flux tube embedded in an otherwise unmagnetized fluid, and Achterberg\cite{Achterberg_1996a,Achterberg_1996b}, who proposed an extension of Spruit's model that satisfies a variational principle. Neither Spruit nor Achterberg establish a precise relationship between their approximate flux tube solutions and true solutions of the underlying MHD system; instead they appeal to the physical plausibility of the so-called slender tube approximation\cite{Achterberg_1996a}. More recently, DeForest and Kankelborg\cite{DeForest_2007} proposed a flux-tube discretization for force-free (i.e. zero pressure) MHD equilibria, which has a public numerical implementation known as FLUX. FLUX has been used to simulate various solar-physical phenomena, such as reconnectionless initiation of coronal mass ejections\cite{Rachmeler_2009}. As with the Achterberg model, the fluxon model of DeForest and Kankelborg is formulated as an approximation of the underlying MHD equilibrium equations. As we will show, SPMHS distinguishes itself from previous flux tube models in various ways. In contrast to previous dynamical flux tube models, SPMHS applies to static MHD equilibria, allows for multiple interacting flux tubes, and is manifestly Hamiltonian. In contrast to DeForest and Kankelborg's fluxon equilibrium model, SPMHS allows for non-zero pressure and enjoys a Hamiltonian dynamical systems formulation. In contrast to both Achterberg and fluxons, SPMHS enjoys a precise mathematical theory relating SPMHS solutions to true solutions of (regularized) MHS equations. (See Theorem \ref{SPMHS_Hamilton} below.)

We formulate SPMHS precisely as follows. An SPMHS particle comprises the particle's position $\bm{q} = (q^x,q^y)\in D$; its momentum $\bm{p} = (p_x,p_y)\in \mathbb{R}^2$; its toroidal flux $S_0^*\in\mathbb{R}$; and its pressure flux $T_0^*\in\mathbb{R}$ (product of pressure with toroidal flux). The single-particle phase space is therefore $T^*D\times\mathbb{R}^2\ni (\bm{q},\bm{p},S_0^*,T_0^*)$. The momentum of a particle on the boundary $\partial D$ is only defined modulo addition of vectors normal to $\partial D$, thereby allowing for reflecting boundary conditions that correspond to the no-outflow condition $\bm{B}\cdot \bm{n} = 0$. We define dynamics for an ensemble of $N$ SPMHS particles using the theory of collectivisation\cite{Guillemin_1980} due to Guillemin and Sternberg. According to their theory, if $\psi:Z_0\rightarrow Z$ is any Poisson map between Poisson manifolds $Z_0,Z$ and $\mathcal{H}:Z\rightarrow\mathbb{R}$ is a possibly time-dependent Hamiltonian on $Z$ then the image under $\psi$ of any solution of Hamilton's equation on $Z_0$ with \textbf{collective Hamiltonian} $h = \mathcal{H}\circ\psi$ is a solution of Hamilton's equation on $Z$ with Hamiltonian $\mathcal{H}$. In the present setting $Z = \mathfrak{p}^*$ is the Lie theoretic phase space for MHS spatial dynamics described in Section \ref{sec:lie_poisson}, $Z_0 = (T^*D\times\mathbb{R})^N$ is the $N$-particle phase space, and $\psi = \Gamma$ is the Poisson map given below in Theorem \ref{thm:fleck_poisson_map}. Since the spatial dynamics Hamiltonian $\mathcal{H}_\zeta$ on $\mathfrak{p}^*$ identified in Section \ref{hamiltonian_formulation} is ill-defined on the image of the Poisson map $\Gamma$, we need to regularize $\mathcal{H}_\zeta$ at small scales before constructing the collective Hamiltonian. Theorem \ref{regularized_hamiltonian_equations} performs the regularization by specifying a regularized Hamiltonian $\mathcal{H}_\zeta^\alpha$ and its corresponding regularized spatial dynamics equations. Here $\alpha\ll 1$ denotes a smoothing length scale familiar from conventional SPH theory. The Hamiltonian $\mathcal{H}_\zeta^\alpha$ and the Poisson map $\Gamma$ give rise to the SPMHS Hamiltonian by collectivisation, $h_\zeta^\alpha = \mathcal{H}_\zeta^\alpha\circ \Gamma$. This ensures the image under $\Gamma$ of any solution of the SPMHS equations described in Theorem \ref{SPMHS_Hamilton} is an exact solution of the regularized MHS spatial dynamics equations defined by $\mathcal{H}_\zeta^\alpha$. Provided the smoothing length $\alpha$ is smaller than any physical length scale within the purview of the ideal MHD model, these exact solutions of the regularized MHS system should be understood as prescribing viable MHD equilibria.

The following Theorem identifies a structure-preserving mapping $\Gamma$ that sends the phase space for $N$ SPMHS particles into the phase space for MHS spatial dynamics. This bridge between particles and fields comprises the first crucial ingredient in building a satisfactory SPMHS theory. The proof presented here that $\Gamma$ is a Poisson map proceeds by brute force. However, it is also possible to construct a more conceptually satisfying proof using ideas developed by Scovel-Weinstein\cite{SW_1994} in the context of the Vlasov-Poisson system. In particular, the alternate proof provides a roadmap for building an SPMHS theory that uses particles with internal degrees of freedom to obtain a richer class of particle-like equilibria. We will investigate this topic in future publications.
\begin{theorem}\label{thm:fleck_poisson_map}
Fix a positive integer $N$ and nonzero real constants $W_S,W_T$. The mapping $\Gamma: (T^*D\times \mathbb{R}^2)^N\rightarrow\mathfrak{p}^*$ given by
\begin{align*}
\begin{pmatrix} \bm{q}^1,\bm{p}^1,S^{*1}_0,T^{*1}_0\\ \vdots \\ \bm{q}^N, \bm{p}^N,S^{*N}_0,T^{*N}_0\end{pmatrix}\mapsto \begin{pmatrix} \sum_{i=1}^N \bm{p}^i\otimes \delta(\bm{x} -\bm{q}^i)\,dx\wedge dy\\
W_S\,\sum_{i=1}^N S_0^{*i}\,\delta(\bm{x} -\bm{q}^i)\,dx\wedge dy\\ W_T\,\sum_{i=1}^N T_0^{*i}\,\delta(\bm{x} -\bm{q}^i)\,dx\wedge dy \end{pmatrix} 
\end{align*}
is a Poisson map. Here $T^*D\times\mathbb{R}^2$ is regarded as the product of the canonical symplectic manifold $T^*D$ with the trivial Poisson manifold $\mathbb{R}^2$, and $\mathfrak{p}^*$ is endowed with the Lie-Poisson structure discussed in Section \ref{sec:lie_poisson}.
\end{theorem}

\begin{proof}
Denote a generic point in the set $T^*D\times\mathbb{R}^2$ using the symbol $(\bm{q},\bm{p},S_0^*,T_0^*)\in T^*D\times\mathbb{R}^2$. Let $W_S,W_T$ be non-zero real constants. First we show directly that the mapping $\gamma: T^*D\times\mathbb{R}^2\rightarrow \mathfrak{p}^*$ given by
\begin{align*}
\gamma(\bm{q},\bm{p},S_0^*,T_0^*)= \begin{pmatrix} \,\bm{p}\otimes \delta(\bm{x} - \bm{q})dx\wedge dy \\W_S\, S_0^*\,\delta(\bm{x} - \bm{q})\,dx\wedge dy \\W_T\, T_0^*\,\delta(\bm{x} - \bm{q})\,dx\wedge dy, \end{pmatrix} 
\end{align*}
is a Poisson map. The Poisson bracket between functions $f,g:T^*D\times\mathbb{R}^2\rightarrow\mathbb{R}$ is given by
\begin{align*}
\{f,g\}(\bm{q},\bm{p},S_0^*,T_0^*)=\partial_{\bm{q}}f\cdot \partial_{\bm{p}}g - \partial_{\bm{q}}g\cdot \partial_{\bm{p}}f,
\end{align*}
while the Poisson bracket between functions $F,G:\mathfrak{p}^*\rightarrow\mathbb{R}$ is given by
\begin{align*}
\{F,G\}_{\mathfrak{p}^*}(X^*) & = -\int_D \bm{U}^* \cdot \left[\frac{\delta F}{\delta \bm{U}^*},\frac{\delta G}{\delta \bm{U}^*}\right]\\
&- \int_D S^*\,\left(\mathcal{L}_{\delta F/\delta \bm{U}^*}\frac{\delta G}{\delta S^*}-\mathcal{L}_{\delta G/\delta \bm{U}^*}\frac{\delta F}{\delta S^*}\right)- \int_D T^*\,\left(\mathcal{L}_{\delta F/\delta \bm{U}^*}\frac{\delta G}{\delta T^*}-\mathcal{L}_{\delta G/\delta \bm{U}^*}\frac{\delta F}{\delta T^*}\right),
\end{align*}
where $X^* = (\bm{U}^*,S^*,T^*)\in\mathfrak{p}^*$.

Notice first that for $F:\mathfrak{p}^*\rightarrow\mathbb{R}$ and $(\delta\bm{q},\delta\bm{p},\delta S_0^*,\delta T_0^*)\in \mathbb{R}^4$ we have
\begin{align*}
\frac{d}{d\lambda}\gamma^*F(\bm{q}+\lambda\,\delta\bm{q},\bm{p} + \lambda\,\delta\bm{p},S_0^*+\lambda\,\delta S_0^*,T_0^* +\lambda\,\delta T_0^*) &=  \,\delta \bm{p}\cdot \frac{\delta F}{\delta\bm{U}^*}(\bm{q}) + W_S\,\delta S_0^*\,\frac{\delta F}{\delta S^*}(\bm{q}) + W_T\,\delta T_0^*\,\frac{\delta F}{\delta T^*}(\bm{q})\\
&+ \delta\bm{q}\cdot \left( \,\partial_{\bm{q}}\frac{\delta F}{\delta\bm{U}^*}\cdot \bm{p} + W_S\,\partial_{\bm{q}}\frac{\delta F}{\delta S^*}\,S_0^* + W_T\,\partial_{\bm{q}}\frac{\delta F}{\delta T^*}\,T_0^*\right).
\end{align*}
The partial derivatives of $\gamma^*F$ are therefore
\begin{align*}
\partial_{\bm{q}}\gamma^*F & =  \,\partial_{\bm{q}}\frac{\delta F}{\delta\bm{U}^*}\cdot \bm{p} + W_S\,\partial_{\bm{q}}\frac{\delta F}{\delta S^*}\,S_0^* + W_T\,\partial_{\bm{q}}\frac{\delta F}{\delta T^*}\,T_0^*,\quad \partial_{\bm{p}}\gamma^*F  =  \,\frac{\delta F}{\delta\bm{U}^*}(\bm{q}),\\
&\quad \partial_{S_0^*}\gamma^*F = W_S\,\frac{\delta F}{\delta S^*}(\bm{q}),\quad \partial_{T_0^*}\gamma^*F  =W_T\, \frac{\delta F}{\delta T^*}(\bm{q}).
\end{align*}
Note that if $\bm{q}\in \partial D$ then $\delta\bm{q}$ is tangent to $\partial D$, implying that $\partial_{\bm{q}}\gamma^*F$ is only defined modulo vectors normal to $\partial D$.
These formulas imply the Poisson bracket $\{\gamma^*F,\gamma^*G\}$ is given by
\begin{align}
\{\gamma^*F,\gamma^*G\} &= \frac{\delta G}{\delta\bm{U}^*}(\bm{q})\cdot\left(\partial_{\bm{q}}\frac{\delta F}{\delta\bm{U}^*}\cdot \bm{p} + W_S\,\partial_{\bm{q}}\frac{\delta F}{\delta S^*}\,S_0^* + W_T\,\partial_{\bm{q}}\frac{\delta F}{\delta T^*}\,T_0^*\right)\nonumber\\
&-\frac{\delta F}{\delta\bm{U}^*}(\bm{q})\cdot\left(\partial_{\bm{q}}\frac{\delta G}{\delta\bm{U}^*}\cdot \bm{p} + W_S\,\partial_{\bm{q}}\frac{\delta G}{\delta S^*}\,S_0^* + W_T\,\partial_{\bm{q}}\frac{\delta G}{\delta T^*}\,T_0^*\right)\nonumber\\
& = -\bm{p}\cdot \left[\frac{\delta F}{\delta\bm{U}^*},\frac{\delta G}{\delta\bm{U}^*}\right](\bm{q}) -W_S\, S_0^*\,\left(\mathcal{L}_{\delta F/\delta\bm{U}^*}\frac{\delta G}{\delta S^*}-\mathcal{L}_{\delta G/\delta\bm{U}^*}\frac{\delta F}{\delta S^*}\right)(\bm{q})\nonumber\\
&-W_T\, T_0^*\,\left(\mathcal{L}_{\delta F/\delta\bm{U}^*}\frac{\delta G}{\delta T^*}-\mathcal{L}_{\delta G/\delta\bm{U}^*}\frac{\delta F}{\delta T^*}\right)(\bm{q}).\label{pb_check_thm}
\end{align}
Notice that only the tangential components of $\partial_{\bm{q}}\gamma^*F$ and $\partial_{\bm{q}}\gamma^*G$ enter into this expression when $\bm{q}\in\partial D$ because $\delta F/\delta \bm{U}^*$ is by definition tangent to $\partial D$. The formula \eqref{pb_check_thm} needs to be compared with $\gamma^*\{F,G\}_{\mathfrak{p}^*}$, which is
\begin{align*}
&\gamma^*\{F,G\}_{\mathfrak{p}^*} = - \,\int_D \bigg(\bm{p}\otimes \delta(\bm{x} - \bm{q})\,dx\wedge dy\bigg)\cdot \left[\frac{\delta F}{\delta \bm{U}^*},\frac{\delta G}{\delta\bm{U}^*}\right]\\
&-W_S\,\int_D \bigg(S_0^*\,\delta(\bm{x} - \bm{q})\,dx\wedge dy\bigg)\left(\mathcal{L}_{\delta F/\delta \bm{U}^*}\frac{\delta G}{\delta S^*} -\mathcal{L}_{\delta G/\delta \bm{U}^*}\frac{\delta F}{\delta S^*} \right)\\
&-W_T\,\int_D \bigg(T_0^*\,\delta(\bm{x} - \bm{q})\,dx\wedge dy\bigg)\left(\mathcal{L}_{\delta F/\delta \bm{U}^*}\frac{\delta G}{\delta T^*} -\mathcal{L}_{\delta G/\delta \bm{U}^*}\frac{\delta F}{\delta T^*} \right)\\
& = \{\gamma^*F,\gamma^*G\}.
\end{align*}
Thus, $\gamma$ is a Poisson map, as claimed.

To complete the proof, we first observe, as done in Scovel-Weinstein\cite{SW_1994}, that if $\psi:Z\rightarrow \mathfrak{p}^*$ is any Poisson map from a Poisson manifold $Z$ into the Lie Poisson space $\mathfrak{p}^*$ then the sum map $\overline{\psi}:Z^N\rightarrow\mathfrak{p}^*: (z^1,\dots, z^N)\mapsto \psi(z^1) +\dots +\psi(z^N)$ is Poisson by linearity of the Poisson bracket on $\mathfrak{p}^*$. With $Z = T^*D\times\mathbb{R}^2$ and $\psi = \gamma$, this is just the statement of the theorem.

\end{proof}

By way of collectivisation, the Poisson map $\Gamma$ provided by Theorem \ref{thm:fleck_poisson_map} would already be sufficient to define the evolution laws for SPMHS if the spatial dynamics Hamiltonian $\mathcal{H}_\zeta$ was unambiguously defined on the image of $\Gamma$. However, due to the ambiguity of products of $\delta$-functions, the naive collective Hamiltonian $h_\zeta = \mathcal{H}_\zeta\circ \Gamma$ is ill-defined. To alleviate this problem in a physically-plausible manner, the following Theorem provides a regularized spatial dynamics Hamiltonian $\mathcal{H}_\zeta^\alpha$. As will become clear explicitly in Theorem \ref{SPMHS_Hamilton}, the composition $h^\alpha_\zeta = \mathcal{H}_\zeta^\alpha\circ\Gamma$ \emph{is} well-defined. Moreover, for small enough smoothing length $\alpha$, the regularized spatial dynamics equations recorded in Theorem \ref{regularized_hamiltonian_equations} are physically indistinguishable from the unregularized MHS equations within the framework of assumptions underlying ideal MHD. 

For regularization, we employ diffusive smoothing as follows. Let $\Delta = \partial_x^2 + \partial_y^2$ denote the standard Laplacian on $D$, regarded as an operator on either scalar fields $S$ or vector fields $\bm{U}$. Consider the vector and scalar heat equations with absolute boundary conditions\cite{Guerini_Savo_2002},
\begin{align}
\partial_t\bm{U}_t - \Delta\bm{U}_t = 0,\quad  \bm{n}\cdot \bm{U}_t = \text{curl}(\bm{U}_t) = 0\text{ on }\partial D, \quad \partial_t S_t -\Delta S_t = 0,\quad  \bm{n}\cdot \nabla S_t = 0\text{ on }\partial D.\label{smoothing_eqns}
\end{align}
The fundamental solutions $\hat{W}_t$, $\hat{Y}_t$ for vector and scalar fields, respectively, define heat kernels $W_t(\bm{r}_1,\bm{r}_2)$, $Y_t(\bm{r}_1,\bm{r}_2)$ such that $\hat{W}_t\bm{U}(\bm{r}) = \int_D W_t(\bm{r},\overline{\bm{r}})\cdot \bm{U}(\overline{\bm{r}})\,d\overline{\bm{r}}$ and $\hat{Y}_tS(\bm{r}) = \int_D Y_t(\bm{r},\overline{\bm{r}})\,S(\overline{\bm{r}})\,d\overline{\bm{r}}$. Note that $\hat{W}_0=\text{id}$, $\hat{Y}_0=\text{id}$, and $\bm{U}_t=\hat{W}_t\bm{U}$, $S_t = \hat{Y}_tS$ satisfy absolute boundary conditions, regardless of whether they are satisfied by the initial conditions $\bm{U},S$. We now precisely define our smooth operators for vector fields and scalar fields as follows.
\begin{definition}\label{smoothing_def}
Let $\alpha > 0$ denote a length scale. The \textbf{smoothing operators} on vector fields $\bm{U}$ and scalar fields $S$, respectively, are defined according to
\begin{gather}
\hat{\mathcal{W}}_\alpha\bm{U} = \hat{W}_{\alpha^2}\bm{U},\quad \hat{\mathcal{Y}}_\alpha S = \hat{Y}_{\alpha^2}S,
\end{gather}
where $\hat{W}_t$ and $\hat{Y}_t$ denote the fundamental solutions of the heat equations \eqref{smoothing_eqns}.
\end{definition}

The temporal scaling $t\sim\alpha^2$ is motivated by well-known asymptotics for heat kernels $W_t(\bm{r}_1,\bm{r}_2)\sim \exp(-|\bm{r}_1-\bm{r}_2|^2/4t)$. We smooth $1$-form densities $\bm{U}^*$ and $2$-forms $S^*$ using the corresponding adjoints $\hat{\mathcal{W}}_\alpha^*$, $\hat{\mathcal{Y}}_\alpha^*$, that is the linear operators defined by requiring
\begin{align*}
\int_D (\hat{\mathcal{W}}_\alpha^*\bm{U}^*)\cdot \bm{U} = \int_D \bm{U}^*\cdot \hat{\mathcal{W}}_\alpha\bm{U},\quad \int_D (\hat{\mathcal{Y}}_\alpha^*S^*)\,S = \int_D S^*\,\hat{\mathcal{Y}}_\alpha\,S,
\end{align*}
for each $\bm{U},S,\bm{U}^*,S^*$. Regardless of the type of field, we denote the result of smoothing using a subscript, e.g. $\bm{U}^*_\alpha = \hat{\mathcal{W}}_\alpha^*\bm{U}^*$. 

The following Lemma provides the precise form of regularized MHS spatial dynamics equations that underlie SPMHS. 
\begin{lemma}\label{regularized_hamiltonian_equations}
Let $\alpha_1,\alpha_2,\alpha_3$ denote positive real constants and set $\alpha = (\alpha_1,\alpha_2,\alpha_3)$. Given a trajectory $(\bm{U}^*_\zeta,S^*_\zeta,T^*_\zeta)\in\mathfrak{p}^*$, let $\bm{U}^*_{\alpha_1\,\zeta} = \hat{\mathcal{W}}^*_{\alpha_1}\bm{U}^*_\zeta$, $S^*_{\alpha_2\,\zeta} =\hat{\mathcal{Y}}^*_{\alpha_2}S^*_\zeta$, and $T^*_{\alpha_3\,\zeta} = \hat{\mathcal{Y}}^*_{\alpha_3}T^*_\zeta$, where $\hat{\mathcal{W}}_{\alpha_i}$ and $\hat{\mathcal{Y}}_{\alpha_i}$ denote the diffusive smoothing operators introduced in Def. \ref{smoothing_def}. Hamilton's equations on $\mathfrak{p}^*\ni X^* = (\bm{U}^*,S^*,T^*)$ associated with the regularized spatial dynamics Hamiltonian 
\begin{align}
\mathcal{H}_\zeta^\alpha(X^*) & = \int_D \left[\frac{1}{2}\left|\frac{\bm{U}^*_{\alpha_1}}{S^*_{\alpha_2}}\right|^2 + \frac{T^*_{\alpha_3}}{S^*_{\alpha_2}} - \frac{1}{2}\left(\frac{S^*_{\alpha_2}}{\omega_\zeta}\right)^2N_\zeta^2\right]\omega_\zeta - \int_D\bm{U}^*_{\alpha_1}\cdot \bm{N}_\zeta\label{reg_ham_func}
\end{align}
are given by
\begin{align}
\partial_\zeta\bm{U}^*_\zeta + \mathcal{L}_{\bm{v}_\zeta}\bm{U}^*_\zeta  &= d\left(\frac{T^*_\zeta}{S^*_\zeta}\right)\otimes \hat{\mathcal{Y}}_{\alpha_3}\left[\frac{\omega_\zeta}{S^*_{\alpha_2\,\zeta}}\right]S^*_\zeta + d\hat{\mathcal{Y}}_{\alpha_2}\left[\frac{\omega_\zeta}{S^*_{\alpha_2\,\zeta}}\left|\frac{\bm{U}^*_{\alpha_1\,\zeta}}{S^*_{\alpha_2\,\zeta}}\right|^2 + \frac{S^*_{\alpha_2\,\zeta}}{\omega_\zeta}N_\zeta^2\right]\otimes S^*_\zeta\nonumber\\
 &+ d\left\{\hat{\mathcal{Y}}_{\alpha_2}\left[\frac{\omega_\zeta}{S^*_{\alpha_2\,\zeta}}\frac{T^*_{\alpha_3\,\zeta}}{S^*_{\alpha_2\,\zeta}}\right] - \hat{\mathcal{Y}}_{\alpha_3}\left[\frac{\omega_\zeta}{S^*_{\alpha_2\,\zeta}}\right]\frac{T^*_\zeta}{S^*_\zeta}\right\}\otimes S^*_\zeta,\quad \bm{v}_\zeta = \hat{\mathcal{W}}_{\alpha_1}\left[\frac{\omega_\zeta}{S^*_{\alpha_2}}\left(\frac{\bm{U}^*_{\alpha_1}}{S^*_{\alpha_2}}\right)^{\sharp} - \bm{N}_\zeta\right] \label{Ustar_reg}\\
\partial_t S^*_\zeta + \mathcal{L}_{\bm{v}_\zeta}S^*_\zeta  &=0 \label{Sstar_reg}\\
\partial_t T^*_\zeta  +\mathcal{L}_{\bm{v}_\zeta}T^*_\zeta  &= 0.\label{Tstar_reg}
\end{align}
\end{lemma}
\begin{proof}
We begin by recording expressions for the functional derivatives of the regularized spatial dynamics Hamiltonian:
\begin{align*}
\frac{\delta\mathcal{H}_\zeta^\alpha}{\delta\bm{U}^*} & = \hat{\mathcal{W}}_{\alpha_1}\left[\frac{\omega_\zeta}{S^*_{\alpha_2}}\left(\frac{\bm{U}^*_{\alpha_1}}{S^*_{\alpha_2}}\right)^{\sharp} - \bm{N}_\zeta\right]\\
\frac{\delta\mathcal{H}_\zeta^{\alpha}}{\delta S^*} & = -\hat{\mathcal{Y}}_{\alpha_2}\left[\frac{\omega_\zeta}{S^*_{\alpha_2}}\left| \frac{\bm{U}^*_{\alpha_1}}{S^*_{\alpha_2}}\right|^2 + \frac{\omega_\zeta}{S^*_{\alpha_2}}\frac{T^*_{\alpha_3}}{S^*_{\alpha_2}} + \frac{S^*_{\alpha_2}}{\omega_\zeta}N_\zeta^2\right]\\
\frac{\delta\mathcal{H}_\zeta^{\alpha}}{\delta T^*} & = \hat{\mathcal{Y}}_{\alpha_3}\left[\frac{\omega_\zeta}{S^*_{\alpha_2}}\right]
\end{align*}
Notice that the vector field $\delta\mathcal{H}_\zeta^{\alpha}/\delta \bm{U}^*$ is automatically tangent to $\partial D$ because the diffusive smoothing operator $\hat{\mathcal{W}}_{\alpha_1}$ is defined using the absolute boundary conditions summarized in Eq.\,\eqref{smoothing_eqns}. We may therefore use the general form of Hamilton's equations on $\mathfrak{p}^*$ given previously in the proof of Theorem \ref{hamiltonian_formulation}, Eqs.\,\eqref{ham_momentum}-\eqref{ham_T}. A simple direct calculation leads to the desired evolution equations.
\end{proof}

\begin{definition}
The \textbf{SPMHS Hamiltonian} is the collective Hamiltonian $h_\zeta^\alpha = \mathcal{H}^\alpha_\zeta\circ\Gamma$ on $N$-particle phase space $(T^*D\times \mathbb{R}^2)^N$ associated with the smoothed spatial dynamics Hamiltonian \eqref{reg_ham_func} and the Poisson map $\Gamma$ described in Theorem \ref{thm:fleck_poisson_map}. Denoting points in $N$-particle phase space using the shorthand $(\underline{\bm{p}},\underline{\bm{q}},\underline{S}^*_0,\underline{T}^*_0)\in (T^*D\times \mathbb{R}^2)^N$, with 
\begin{align*}
\underline{\bm{p}} = (\bm{p}^1,\dots,\bm{p}^N),\quad\underline{\bm{q}} = (\bm{q}^1,\dots,\bm{q}^N),\quad \underline{S}^*_0 = (S^{*\,1}_0,\dots,S^{*\,N}_0),\quad \underline{T}^*_0 = (T^{*\,1}_0,\dots,T^{*\,N}_0),
\end{align*}
the SPMHS Hamiltonian is given explicitly by
\begin{align}
h^\alpha_\zeta(\underline{\bm{q}},\underline{\bm{p}},\underline{S}^*_0,\underline{T}^*_0) & = \frac{1}{2}\bigg(\underline{\bm{p}}-\mathbb{A}_\zeta(\underline{\bm{q}},\underline{S}^*_0)\bigg)^T\mathbb{M}_\zeta(\underline{\bm{q}},\underline{S}^*_0)\bigg(\underline{\bm{p}} - \mathbb{A}_\zeta(\underline{\bm{q}},\underline{S}^*_0)\bigg)+ \mathcal{V}_\zeta(\underline{\bm{q}},\underline{S}^*_0,\underline{T}^*_0),\label{reg_collective_hamiltonian}
\end{align}
where
\begin{align*}
\mathbb{M}_\zeta^{(a_1,\nu_1)\,(a_2\,\nu_2)}(\underline{\bm{q}},\underline{S}^*_0) & = \int_D \frac{[W_{\alpha_1^2}]^{\nu_1}_{\mu_1}(\bm{q}^{a_1},\bm{x}) h_\zeta^{\mu_1\,\mu_2}(\bm{x})\,[W_{\alpha_1^2}]^{\nu_2}_{\mu_2}(\bm{q}^{a_2},\bm{x})  }{W_S^2\sum_{c_1\,c_2} S_0^{*c_1}\,S_0^{*c_2} Y_{\alpha_2^2}(\bm{q}^{c_1},\bm{x})\,Y_{\alpha_2^2}(\bm{q}^{c_2},\bm{x}) }\,N_\zeta(\bm{x})\,\sqrt{\text{det}\,h_\zeta}d\bm{x}\\
\mathbb{N}_\zeta^{(a,\nu)}(\underline{\bm{q}}) & = \int_D [W_{\alpha_1^2}]^\nu_\mu(\bm{q}^a,\bm{x})\,N_\zeta^\mu(\bm{x})\,d\bm{x}\\
\mathbb{A}_\zeta(\underline{\bm{q}},\underline{S}^*_0) & = \mathbb{M}^{-1}_\zeta(\underline{\bm{q}},\underline{S}^*_0)\mathbb{N}_\zeta(\underline{\bm{q}})\\
\mathcal{V}_\zeta(\underline{\bm{q}},\underline{S}^*_0,\underline{T}^*_0) & =  \frac{W_T}{W_S}\int_D \frac{\sum_{c} T_0^{*\,c}Y_{\alpha_3^2}(\bm{q}^c,\bm{x})}{\sum_{\overline{c}}S_0^{*\overline{c}}\,Y_{\alpha_2^2}(\bm{q}^{\overline{c}},\bm{x})}\,N_\zeta(\bm{x})\,\sqrt{\text{det}\,h_\zeta}d\bm{x} -\frac{1}{2}\mathbb{N}^T_\zeta(\underline{\bm{q}})\mathbb{M}_\zeta^{-1}(\underline{\bm{q}},\underline{S}^*_0)\mathbb{N}_\zeta(\underline{\bm{q}}) \\
&- \frac{1}{2}W_S^2\int_D \sum_{c\,\overline{c}}S_0^{*c}S_0^{*\overline{c}}\,Y_{\alpha_2^2}(\bm{q}^c,\bm{x})\,Y_{\alpha_2^2}(\bm{q}^{\overline{c}},\bm{x})\,\frac{N_\zeta(\bm{x})}{\sqrt{\text{det}\,h_\zeta}}\,d\bm{x} 
\end{align*}
The \textbf{SPMHS equations of motion} are Hamilton's equations on $(T^*D\times\mathbb{R}^2)^N$ associated with the SPMHS Hamiltonian. They are given explicitly by
\begin{align*}
\partial_\zeta\underline{\bm{q}} = \partial_{\underline{\bm{p}}}h_\zeta^\alpha,\quad \partial_\zeta\underline{\bm{p}} = -\partial_{\underline{\bm{q}}}h_\zeta^{\alpha},\quad \partial_{\zeta}\underline{S}^*_0 = 0,\quad \partial_{\zeta}\underline{T}^*_0 = 0.
\end{align*}
Whenever a particle encounters the spatial boundary $\partial D$, it is reflected by reversing the normal component of the particle's momentum. 
\end{definition}

\begin{theorem}\label{SPMHS_Hamilton}
Let $\Gamma$ be the Poisson map described in Theorem \ref{thm:fleck_poisson_map}. Let $(\underline{\bm{q}}(\zeta),\underline{\bm{p}}(\zeta),\underline{S}^*_0,\underline{T}^*_0)$ be any solution of the SPMHS equations of motion. Let $I =(\zeta_1,\zeta_2)$ be any time interval over which no particle encounters $\partial D$. The image of $(\underline{\bm{q}}(\zeta),\underline{\bm{p}}(\zeta),\underline{S}^*_0,\underline{T}^*_0)$ under $\Gamma$ is a solution of the regularized spatial dynamics equations given in Lemma \ref{regularized_hamiltonian_equations} on the time interval $ I$.
\end{theorem}
\begin{proof}
After noticing that the regularized collective Hamiltonian \eqref{reg_collective_hamiltonian} is well-defined, and recalling from Theorem \ref{thm:fleck_poisson_map} that $\Gamma:(T^*D\times\mathbb{R}^2)\rightarrow\mathfrak{p}^*$ is a Poisson map, this is an immediate corollary of the Guillemin-Sternberg\cite{Guillemin_1980} collectivisation theory
\end{proof}



\section*{Discussion\label{sec:discussion}}

This work lays the theoretical foundation for studying magnetohydrostatics (MHS) as a spatial dynamical system, where a toroidal angle $\zeta$ is treated as ``time". It reveals the variational and Hamiltonian structure of the spatial dynamics equations, and uses the latter to formulate a theory of smoothed-particle magnetohydrostatics (SPMHS). Notably, SPMHS comprises a finite-dimensional Hamiltonian system whose periodic solutions correspond to exact solutions of a physically-plausible short-scale regularization of MHS. Finite-dimensionality side-steps traditional analytical difficulties leading to Grad's conjecture\cite{Grad_1967} on non-existence of three-dimensional equilibria. The spatial dynamics perspective also leads to an appealing resolution of the mixed ellipticity-hyperbolicity of the MHS equations in large-aspect ratio domains. Large aspect ratio implies the spatial dynamics equations comprise a fast-slow system with a normally-hyperbolic slow manifold\cite{Fenichel_1979}. Motion on the slow manifold captures the hyperbolic part of the equations, while fast motion normal to the slow manifold corresponds to the elliptic part. This means that approximate three-dimensional equilibria can be studied using the (hyperbolic) formal slow manifold reduction of spatial dyanmics, thereby side-stepping analytical issues associated with mixed type PDE.

The spatial dynamics approach to MHS represents an opportunity to improve computations of stellarator equilibrium configurations. Today, workhorse stellarator equilibrium solvers, like VMEC, DESC, and NSTAB, run with fairly coarse resolution in order to cope with sheet-like singularities that tend to appear near bad rational surfaces. This implies the solvers are not finding true solutions of the MHS equations, but rather solutions of some regularized problem that has yet to be completely characterized. This gap between the regularized problem solvers actually address and the standard MHS problem raises questions about the physical nature of stellarator equilibria that are difficult to address without a useful characteriztion of the regularization. In contrast, in the SPMHS approach to equilibrium computation, controlled, physically-plausible regularization occurs before formulation of any numerical scheme. Finite-dimensionality of SPMHS implies that it should therefore be possible to compute well-resolved solutions of the regularized equations. \emph{A priori} characterization of the regularization implies that the implications on stellarator theory of using solutions of the regularized equations in place of MHS solutions can be investigated directly. Similar possibilities offered by the slow manifold reduced MHS equations need to be investigated further. But the limit slow equations, which are equivalent to 2D Euler with advected scalar tracer, certainly enjoy a fairly well-understood theory at both the discrete and continuum levels.

The promise of the SPMHS approach should be contrasted with the multi-region relaxed MHD (MRxMHD) method implemented in the equilibrium solver SPEC\cite{Hudson_2012}. SPEC aims to compute genuine weak solutions of the MHS equations in the sense described precisely by Bruno-Laurence\cite{BL_1996}. But because the Bruno-Laurence theory only establishes existence of such solutions when deviations from axisymmetry are small, a good convergence theory for SPEC when applied to stellarators (which are usually far from being axisymmetric) remains out of reach\cite{Qu_2021}. In principle, and as with other stellarator equilibrium solvers, there is a gap between what SPEC actually computes and the MHS model it aims to solve. Moreover, best practices for applying SPEC as a workhorse stellarator equilibrium solver remain in development. In contrast, existence theory for SPMHS and convergence theory for time discretization of SPMHS promise to be comparatively simple because SPMHS is governed by ordinary differential equations (ODE) instead of PDE. This suggets feasibility of closing the model-solver gap for SPMHS. Practical application of SPMHS also stands to benefit from decades of experience with particle-based simulation methods\cite{Dawson_1983,Squire_2012,Evstatiev_2013,Kraus_2017,Burby_fdk_2017,Xiao_2018,Li_2019,Pinto_2022} within the field of plasma physics. Overall, both MRxMHD and SPMHS warrant further serious investigation.

Exploring the nature of the slow manifold reduced spatial dynamics equations in large-aspect-ratio domains represents a rich area of future research. As described in Section \ref{sec:fs_formulation}, asymptotic expansions provide one avenue for studying the slow manifold. For any practical application, such expansions need to be truncated. Performing this truncation while preserving the Hamiltonian structure of spatial dynamics described in Section \ref{sec:lie_poisson} stands as an open problem. This might be approachable by combining symplectic methods with symmetry reduction, as done previously for dynamical MHD\cite{Burby_TF_2017}, weakly-relativistic Vlasov-Maxwell\cite{Miloshevich_2021}, and Lorentz loop dynamics\cite{Burby_loops_2020}. A more direct approach would involve studying the slow manifold as a Poisson transversal\cite{Weinstein_1983}. Data-driven methods offer another avenue for investigating the slow manifold. Library-based learning methods such as weak-form SINDy\cite{Messenger_wSINDy_2021} could be used to compute terms in the slow manifold asymptotic expansion, thereby avoiding human error that inevitably creeps into manual higher-order perturbation theory. Alternatively, more expressive machine learning models such as neural networks could be used to learn the slow manifold from data, even when the large-aspect-ratio ordering parameter $\epsilon$ is not particularly small. It would be especially exciting to learn the slow manifold while preserving its Hamiltonian structure using Hamiltonian graph neural networks\cite{Bishnoi_2023}.

The SPMHS theory developed here is modeled on the simplest version of the Lie algebra reduction method\cite{SW_1994} introduced by Scovel-Weinstein. Applying the full Scovel-Weinstein machinery to spatial dynamics would result in a variant of SPMHS where each particle acquires additional shape degrees of freedom. Exploring the implications of this higher-order SPMHS theory warrants further investigation. At a minimum, it would represent a method for combining SPMHS particles together, as discussed by Holm-Jacobs\cite{Holm_2017} in the context of multipole vortex blobs. SPMHS represents a solution the MHS equations as a collection of flux tubes. Scovel-Weinstein theory may also be used to decompose a solution into a collection of flux surfaces. The theory would proceed along similar lines to Burby's application\cite{Burby_moments_2023} of Scovel-Weinstein theory to fluid moments. The slow manifold reduced spatial dynamics equations, being equivalent at leading-order to 2D Euler with advected tracer, should admit a Lie-Poisson Hamiltonian structure. If this can be confirmed, the structure can be used to develop SPMHS on the slow manifold.




\section*{Acknowledgements}

Research presented in this article was supported by the U.S. Department of Energy (DOE), the Office of Science and the
Office of Advanced Scientific Computing Research (ASCR). Specifically, we acknowledge funding support from ASCR for
DOE-FOA-2493 “Data-intensive scientific machine learning and analysis”. 








\end{document}